\documentclass[11p,reqno]{amsart}
\usepackage[colorlinks=true, pdfstartview=FitV, linkcolor=blue, citecolor=blue, urlcolor=blue]{hyperref}
\usepackage[usenames,dvipsnames]{xcolor}
\usepackage{latexsym}
\usepackage{amssymb}
\usepackage{amscd}
\usepackage{bm}
\usepackage{amsmath,amsthm}
\usepackage{setspace}
\usepackage{amssymb,amsthm,amsmath,mathrsfs,bm}
\usepackage{graphicx}
\usepackage{tabularx}
\usepackage{subfigure}
\usepackage{appendix}
\usepackage{enumerate}
\usepackage{cite}
\usepackage{tikz}
\usetikzlibrary{cd}

\newcommand{\be}{\begin{eqnarray}}
\newcommand{\ee}{\end{eqnarray}}
\newcommand{\bez}{\begin{eqnarray*}}
\newcommand{\eez}{\end{eqnarray*}}

\numberwithin{equation}{section}

\numberwithin{equation}{section}
\newtheorem{theorem}{Theorem}[section]
\newtheorem{lemma}[theorem]{Lemma}
\newtheorem{coro}[theorem]{Corollary}

\theoremstyle{definition}

\newtheorem{remark}[theorem]{Remark}

\allowdisplaybreaks[4]

\begin{document}

\title[]{Hankel determinants from  quadratic orthogonal pairs for  hyperelliptic functions and their applications}

\date{}

\author{Xiang-Ke Chang}
\address{SKLMS \& ICMSEC, Academy of Mathematics and Systems Science, Chinese Academy of Sciences, P.O. Box 2719, Beijing 100190, PR China; and School of Mathematical Sciences, University of Chinese Academy of Sciences, Beijing 100049, PR China.}
\email{changxk@lsec.cc.ac.cn}

\author{Jiyuan Liu}
\address{ SKLMS \& ICMSEC, Academy of Mathematics and Systems Science, Chinese Academy of Sciences, P.O. Box 2719, Beijing 100190, PR China; and School of Mathematical Sciences, University of Chinese Academy of Sciences, Beijing 100049, PR China.}
\email{liujiyuan@amss.ac.cn}

\begin{abstract}
As argued by Hone in the paper [Commun. Pure Appl. Math., 74(11):2310--2347, 2021], a ``mismatch"  problem remained  unsolved while he was investigating continued fraction expansions and Hankel determinants from hyperelliptic curves. In this paper, by introducing a new notion called quadratic orthogonal pairs for hyperelliptic functions, we resolve the corresponding problem. As further applications, we give a thorough treatment of the initial value problems for two discrete integrable systems, i.e. the bilateral Somos-4 and Somos-5 recurrences.

\end{abstract}


\maketitle
\tableofcontents

\section{Introduction}
The so-called Somos-4 recurrence is a three-term discrete quadratic equation of fourth order
\begin{equation}
	S_{n+4}S_{n} =\alpha S_{n+3}S_{n+1} + \beta S_{n+2}^2, \label{somoseq4}
\end{equation}
where $\alpha$ and $\beta$ are constant parameters, while the Somos-5 recurrence is defined as a discrete quadratic equation
\begin{equation}
	S^*_{n+5}S^*_{n} =\alpha S^*_{n+4}S^*_{n+1} + \beta S^*_{n+2}S^*_{n+1}.\label{somoseq5}
\end{equation}
They can both be interpreted as discrete integrable systems \cite{hone2005elliptic,hone2021,hone2007sigma}. In fact, they could be regarded as reductions of the discrete Hirota equation (also referred to as the bilinear discrete KP, or the octahedron recurrence)  \cite{hone2017reductions,krichever1998elliptic,speyer2007perfect}, which is one of the most important discrete integrable systems. Along with their high-order analogues, these two recurrences have attracted considerable attention.

As for the Somos-4 recurrence, an early and remarkable observation is that, if the initial values are set to be $(S_0,S_1,S_2,S_3)= (1,1,1,1)$, then the \textit{bilateral} Somos-4 sequence $\{S_n\}_{n \in \mathbb{Z}}$ produced by the recurrence \eqref{somoseq4} with $\alpha=\beta=1$ consists of integers\cite{somos1989problem,malouf1992integer,Gale1991strange}
\begin{equation}
	\ldots,314,59,23,7,3,2,1,1,1,1,2,3,7,23,59,314,\ldots .  \label{somos4sequence}
\end{equation}
One efficient and impressive way to prove this fact is to fit the iteration \eqref{somoseq4} into a big picture called cluster algebras, which was drawn by Fomin and Zelevinsky. In their theory of cluster algebras \cite{fomin2002laurent,fomin2002cluster}, the Somos-4 recurrence is interpreted as a path of a 4 regular tree $\mathbb{T}_{4}$ driven by a quiver with 4 vertices. In fact, they proved that, for a large class of cluster algebras of geometric type,  the cluster variables exhibit the Laurent property, that is, any cluster variable is a Laurent polynomial when written as a rational function in terms of any other cluster. As special cases, certain Somos recurrences admit the Laurent property, from which, the fact that the sequence \eqref{somos4sequence} produced by the Somos-4 recurrence \eqref{somoseq4} are all integers immediately follows as a straightforward corollary.

In addition to the Laurent phenomenon closely related to the theory of cluster algebras, Somos recurrences also exhibit some other distinctive features. As mentioned above, certain Somos recurrences and their higher-order analogues can be thought of discrete integrable systems, serving as integrable symplectic mappings \cite{hone2018some}, as reductions of the discrete Hirota equation \cite{hone2017reductions,speyer2007perfect} or Miwa equation (also known as the bilinear discrete BKP, or the cube recurrence) \cite{fedorov2016sigma,carroll2004cube}. Somos sequences naturally appear in some problems in number theory \cite{ward1948memoir,robinson1992periodicity,everest2005introduction,everest2003recurrence,hone2022heron}. 
They may also arise in solvable models within the realms of statistical mechanics and quantum field theory, such as the hard hexagon model \cite{quispel1988integrable}, dimer models and quiver gauge theory \cite{eager2012colored}.

It is interesting that general solutions to  certain Somos recurrences  can be derived. It has been shown that certain Somos recurrences admit general explicit solutions in terms of sigma functions associated with hyperelliptic curves \cite{hone2005elliptic,hone2007sigma,hone2010analytic,swart2003elliptic,fedorov2016sigma}.
In particular, the general solutions to the Somos-4  and Somos-5 recurrences can be constructed in terms of Weierstrass sigma function \cite{hone2005elliptic,hone2007sigma}, and can also be expressed using closed forms in terms of Hankel determinants \cite{chang2012conjecture,chang2015hankel,hone2023casting,hone2021,xin2009proof,hone2023family} with the matrix entries from elliptic curves. For some other discussions on the related topics, see e.g. \cite{han2025,ovsienko2025continued,van2006recurrence,wang2024sufficient}.

The Somos-4 recurrence enjoys a structure of Hankel determinants. Recall that a sequence of Hankel determinants $\{H_{n}\}_{n \in \mathbb{N}}$ is defined according to
\begin{equation*}
	H_{n} :=
	\left\{
	\begin{aligned}
		&\det(s_{i+j-2})_{i,j= 1,\ldots,n}, \quad & n \ge 1,\\ 
		&1, \quad & n = 0,
	\end{aligned}        
	\right. 
\end{equation*}
for the matrix entries $\{s_n\}_{n \in \mathbb{N}}$, which, in some sense, can be regarded as a transformation from one sequence $\{s_n\}_{n \in \mathbb{N}}$ to another sequence $\{H_n\}_{n \in \mathbb{N}}$. It was originally conjectured by Somos \cite{somos1989problem}, later proved by Xin \cite{xin2009proof} that, when the sequence $\{s_n\}_{n \in \mathbb{N}}$ is chosen to be the coefficients of the Taylor expansion of a certain function $G(z)$
\begin{equation*}
	G (z):= \frac{y(z)}{z}=\frac{\frac{1}{2}+\sqrt{\frac{1}{4}-z+z^3}}{z} - 1 = \sum_{n \ge 1} s_{n-1}z^n
\end{equation*}
obtained from the function $y=y(z)$ satisfying the restriction $y-y^2=z-z^3$,
then the corresponding sequence of Hankel determinants $\{H_{n}\}_{n \in\mathbb{N}}$ exactly gives a half of the Somos-4 sequence 
\begin{equation*}
	1,1,2,3,7,23,59,314,\ldots .
	\end{equation*}	
It was further conjectured by Barry that, if the sequence $\{s_n\}_{n \in \mathbb{N}}$ satisfies a particular convolution recursion \cite{barry2010generalized}, then the corresponding sequence of Hankel determinants $\{H_{n}\}_{n \in \mathbb{N}}$ also gives a half of the Somos sequence with specific $\alpha$ and $\beta$. This conjecture was proved by Hu and one of the authors in \cite{chang2012conjecture}. Later in \cite{chang2015hankel}, together with Xin, they derived the Hankel determinant formulae for the half of the Somos-4 recurrence with arbitrary $\alpha$ and $\beta$. All of these results are devoted to the half of the Somos-4 recurrence.
\par
As for the \textit{bilateral} Somos-4 recurrence \eqref{somoseq4} with $n\in \mathbb{Z}$, at first glance,  Hankel determinants may not be the ideal tool to describe its general solution,  since  it is generally not natural to extend the definition of Hankel determinants from natural numbers to all integers.
To obtain an overall expression of the bilateral Somos-4 sequences, in \cite{hone2021}, Hone made an attempt based on the connections with continued fraction expansions of functions on hyperelliptic curves, which are also related to the geometry of the Jacobians of the curves \cite{adams1980,zannier2019hyper}.
 He expressed the positive sequence and the negative sequence in terms of Hankel determinants independently, and then glue them together to get an overall expression using gauge transformations; see \cite[Remark 4.6]{hone2021}. However, his process seems to break the overall structure, and the relationship between the positive part and the negative part still remain unclear. In fact, as mentioned in  \cite[Remark 4.6]{hone2021}, \textbf{a ``mismatch" problem arises, that is, ``in general one cannot just take $S_n=H_n$ for nonnegative $n$ and $S_n=H_{-n-1}^\dag$ for negative $n$, because there will be a mismatch at the values of $d_0$, $d_1$, and $v_0$ that are left unspecified."}

This paper is mainly motivated by the above mismatch problem  \cite[Remark 4.6]{hone2021}. To tackle it, by digging up the interior symmetries hidden in the continued fraction recursions, we introduce a new concept called \textit{quadratic orthogonal pair} for a pair of generating functions associated with hyperelliptic curves.

More precisely, we start from a concept called ``proper" quadratic field extension over $\mathbb{C}(X)$ and show that the quadratic field extensions induced by hyperelliptic curves are all proper so that the continued fraction theory can apply. Then,  based on the Galois symmetry of a proper quadratic field extension, by introducing the new notion called quadratic orthogonal pairs for hyperelliptic functions, and combining it with the continued fraction theory for the Laurent series field $\mathbb{C}((X^{-1}))$, we obtain the consistent expressions \eqref{tauseq} and \eqref{d}  together with Lemma \ref{coeff} that are defined for all integers, which eventually solves the mismatch problem raised by Hone in \cite[Remark 4.6]{hone2021}. It is noted that this notion also reveals the relation between the two generating functions from which one can get an overal expression of solutions to a   particular nonlinear recurrence. A correspondence between iterations of continued fraction type for generating functions and guage transformations for Somos 4 sequences is also given in Theorem \ref{rela_n-n-1}. 

As further applications of our approach, we derive the general solution to the initial value problem for the bilateral Somos-4 recurrence using Hankel determinants defined on a quadratic orthogonal pair for elliptic curves. The explicit Hankel determinant expressions also allow us to verify the Laurent property of the bilateral  Somos-4 recurrence. In addition, the initial value problem of the bilateral Somos-5 recurrence is also solved using Hankel determinants defined on two quadratic orthogonal pairs. We also show how our approach on quadratic orthogonal pairs is used in generalizing Somos' original conjecture on half of Somos-4 sequence to the bilateral case. In particular, we find a ``symmetric" curve closely related to the fundamental curve $y-y^2=z-z^3$.

Furthermore, the integrability of the continued fraction recursions for hyperelliptic functions is also studied. In particular we show that every hyperelliptic function is naturally associated with a projective transformation with two fixed points and the continued fraction expansion can also be regarded as a projective transformation. To understand the Lax integrability, we find that the compatibility conditions of these two transformations are closely related to the symmetries of the continued fraction expansions.

The paper is organized as follows. In Section \ref{sec:CF}, we present some facts on continued fractions in $\mathbb{C}((X^{-1}))$ with a particular focus on J-fractions and the corresponding Hankel determinant representations. In Section \ref{sec:QOP}, we investigate the quadratic field extensions induced by hyperelliptic curves with a concentration on genus 1.
We demonstrate that every hyperelliptic function is naturally associated with a Lax pair and the compatibility condition is equivalent to a certain restriction on the continued fraction expansions of the hyperelliptic functions. The geometric meaning of the Lax matrices is also clarified. By further introducing pairs of hyperelliptic functions called quadratic orthogonal pairs, we solve the mismatch problem posed by Hone in \cite[Remark 4.6]{hone2021}.  Section \ref{sec:somos} is devoted to further applications of our approach to initial value problems of the bilateral Somos-4 and Somos-5 recurrences, based on which Laurent properties are also discussed. Additional comments on  quadratic orthogonal pairs for hyperelliptic functions with higher genus are collected in Appendix \ref{generalizepairs}.

\section{On J-fractions and Hankel determinants} \label{sec:CF}
In this section, we introduce some facts on continued fractions with a particular focus on J-fractions. The related materials can be found in e.g. \cite{wall1,wall2}.  
\subsection{Continued fraction expansions in \texorpdfstring{$\mathbb{C}((X^{-1}))$}{}} \label{A}
\par
Let $\mathbb{K}$ be a field. For a sequence $\{a_{n}\}_{n\in \mathbb{N}}$ in $\mathbb{K}$, its partial continued fraction $\langle a_0,\ldots,a_n \rangle  $ is defined by
\begin{equation*}
	\langle a_0,\ldots,a_n \rangle := a_0 + \cfrac{1}{a_1 
		+ \cfrac{1}{a_2 
			+ \cfrac{1}{a_3 +\cfrac{1}{ \cdot\cdot\cdot +\frac{1}{a_n}} } } } .
\end{equation*}
In general, the partial continued fraction can be written as a fraction, that is,
\begin{equation*}
	\langle a_0,\ldots,a_n \rangle   = \frac{p_n}{q_n},
\end{equation*}
where $p_n$ and $q_{n}$ are polynomials in $\{a_{n}\}_{n\in \mathbb{N}}$  recursively defined by 
\begin{equation}
	\left\{\begin{matrix}
		p_n= a_n p_{n-1} +p_{n-2} \\
		q_n= a_n q_{n-1} +q_{n-2} 
	\end{matrix}\right.,\quad n \geqslant 0,\label{polyreprensentation}
\end{equation}
with initial values
$$
\left\{\begin{matrix}
	p_{-1}= 1 \\
	q_{-1}= 0
\end{matrix}\right., \qquad 
\left\{\begin{matrix}
	p_{-2}= 0 \\
	q_{-2}= 1
\end{matrix}\right. .
$$

\par
The above procedure can be generalized to function field cases when $\mathbb{K}$ is taken to be a certain function field and $\{a_n\}_{n \in \mathbb{N}}$ to be a series of functions. In particular, if we choose $\{a_n\}_{n \in \mathbb{N}}$ to be a series of linear functions of one variable, then the corresponding continued fractions are said to be J-fractions\cite{wall1}. 
\par
For later use, we now choose
$\mathbb{K} = \mathbb{C}((X^{-1}))$ to be the Laurent series field in one variable over $\mathbb{C}$, and examine the continued fraction expansion of a given function $f = \sum_{n \in \mathbb{Z}} c_n X^n$ in $\mathbb{C}((X^{-1}))$, where $c_n$ are complex numbers satisfying the condition that there exists an integer $N$ such that $c_n = 0$ for all $n \ge N$. Here we make the convention that, throughout the paper, whenever we mention an element in $\mathbb{C}((X^{-1}))$, this condition is always assumed. Under this assumption, the following map is well defined
\begin{equation}
	\begin{aligned}
		\max  :  \mathbb{C}((X^{-1})) &\longrightarrow   \mathbb{Z} \\
		f = \sum_{n \in \mathbb{Z}} c_n X^n & \longmapsto 
		\max\{n \in \mathbb{Z}| f_n \ne 0\} .
	\end{aligned}
	\label{max}
\end{equation}
\par
For elements in $\mathbb{C}((X^{-1}))$, define an operation as follows
\begin{equation*}
	\begin{aligned}
		\left \lfloor \cdot \right \rfloor  :  \mathbb{C}((X^{-1})) &\longrightarrow   \mathbb{C}[X] \\
		f = \sum_{n \in \mathbb{Z}} c_n X^n & \longmapsto    \left \lfloor f \right \rfloor=\sum_{n \ge 0 }c_n X^n .
	\end{aligned}
\end{equation*}
This operation maps an element to its polynomial part so that one has a natural decomposition of elements in $\mathbb{C}((X^{-1}))$
\begin{equation}
	f= \left \lfloor f \right \rfloor + (f-\left \lfloor f \right \rfloor) .\label{decompose}
\end{equation}
Using this decomposition, one can easily construct a continued fraction expression for each element in $\mathbb{C}((X^{-1}))$. More precisely, given  an arbitrary element $f_0 \in \mathbb{C}((X^{-1}))$, decomposing $f_0$ using \eqref{decompose}, we have
\begin{equation*}
	f_0= \left \lfloor f_0 \right \rfloor + (f_0-\left \lfloor f_0 \right \rfloor) =: a_0 + b_0,
\end{equation*}
where we use $a_0$ to denote the polynomial part and $b_0$ to denote the negative series part. If $b_0$ is not zero, then it is invertible in $\mathbb{C}((X^{-1}))$. Thus we have a new element in $\mathbb{C}((X^{-1}))$
\begin{equation*}
	f_1 := b_0^{-1} .
\end{equation*}
Note that the same operation can be applied to $f_1$
\begin{equation*}
	\begin{aligned}
		&f_1= \left \lfloor f_1 \right \rfloor +(f_1-\left \lfloor f_1 \right \rfloor) =: a_1+ b_1,\\
		&f_2 := b_1^{-1},
	\end{aligned}
\end{equation*}
if $b_1\neq0$.
Repeating this operation, we have, for all natural number $n$
\begin{equation*}
	\begin{aligned}
		&f_n= \left \lfloor f_n \right \rfloor +(f_n-\left \lfloor f_n \right \rfloor) =: a_n+ b_n, \\
		&f_{n+1} := b_n^{-1},
	\end{aligned}
\end{equation*}
once $b_n\neq0$.
Thus we have, for any $ 0 \le m \le n$,
\begin{equation*}
	f_m = \langle a_m,a_{m+1},\ldots,a_{n},f_{n+1} \rangle  .
\end{equation*}
If follows from \eqref{polyreprensentation} that $f_m$ can be expressed in terms of a fraction of two polynomials in $\{a_m,\ldots,a_{n},f_{n+1}\}$. In addition, it is clear that, by definition of $\left \lfloor \cdot \right \rfloor$, we have
\begin{equation*}
	\left \lfloor -f \right \rfloor = -\left \lfloor f \right \rfloor .
\end{equation*}
This implies that, if $f$ and $\Tilde{f}$ satisfying $f+ \Tilde{f}=0$ are expanded in continued fractions as
\begin{equation*}
	\begin{aligned}
		f = \langle f_0,f_1,\ldots,f_n,\ldots \rangle, \qquad \Tilde{f} =\langle \Tilde{f}_0,\Tilde{f}_1,\ldots,\Tilde{f}_n,\ldots \rangle,
	\end{aligned}
\end{equation*}
then
\begin{equation*}
	f_n+ \Tilde{f}_n=0,\quad \forall\, n \in \mathbb{N},
\end{equation*}
which is quite different from the continued fraction theory for real numbers.

\subsection{Hankel determinant representations for J-fractions}
\par
As a particular type of continued fractions, J-fractions are known to be useful in dealing with Hankel determinants\cite{krattenthaler2005advanced,wall2}. In general, the quantities in  J-fractions can be represented in terms of Hankel determinants with the coefficients of the corresponding series expansion as elements.

Given an element $f$ in $\mathbb{C}((X^{-1}))$, if the corresponding continued fraction expansion of $f$ is a J-fraction, then $f$ is said to be J-expressible, that is,
$f$ can be expressed as a J-fraction as follows
\begin{equation}
	f = \langle a_0,a_1,a_2,\ldots \rangle, \label{fractions}
\end{equation}
where $a_0 = aX+b$, $a_n= \alpha_n X +\beta_n$, $\alpha_n$ and $\beta_n$ are all complex numbers, $\alpha_n \ne 0$. 
Conversely, if $f$ is J-expressible, then $f$ must have the following form
\begin{equation}
	f = aX+b+ s_0 X^{-1} +s_1 X^{-2}+\cdot\cdot\cdot . \label{series}
\end{equation}
Combining \eqref{series} and \eqref{fractions}, we have
\begin{equation*}
	\begin{aligned}
		aX+b + s_0 X^{-1} +s_1 X^{-2}+\cdot\cdot\cdot  &= \langle a_0,a_1,a_2,\ldots \rangle    \\
		&= a_0 +\langle 0,a_1,a_2,\ldots \rangle   \\
		&= aX + b + \langle 0,a_1,a_2,\ldots \rangle, \\
	\end{aligned}
\end{equation*}
which gives 
\begin{equation}
	\langle 0,\alpha_1X+\beta_1,\alpha_2X+\beta_2,\ldots \rangle   = s_0 X^{-1} +s_1 X^{-2}+\cdot\cdot\cdot .\label{eq}
\end{equation}
This means that there must exist a close relation between the set 
$\{\alpha_n,\beta_n\}_{n \in \mathbb{N}^*}$ and the set $\{s_n\}_{n \in \mathbb{N}}$, where $\mathbb{N}^*$ means the set of all positive integers. Actually, the relations between these two sets are given by the following theorem (see e.g. \cite[Theorem A]{wall1},\cite[Theorem 51.1]{wall2} or \cite[Theorem 29]{krattenthaler2005advanced}). The proof can be achieved in various ways; see e.g. \cite[proof of Theorem A]{wall1} based on the underlying orthogonality or  \cite[Appendix]{hone2021} based on identities for determinants.
\begin{theorem}\label{connection}
	Suppose $\langle 0,\alpha_1X+\beta_1,\alpha_2X+\beta_2,\ldots \rangle   = s_0 X^{-1} +s_1 X^{-2}+\cdot\cdot\cdot $, then, for any $n \in \mathbb{N}^*$, there hold
	\begin{equation*}
		\begin{aligned}
			 - \alpha_{n-1}^{-1}\alpha_{n}^{-1} &= \frac{\Delta_{n}\Delta_{n-2}} {\Delta_{n-1}^2},\\
			 \alpha_n^{-1} \beta_{n}& = \frac{\Delta_{n-1}^*}{\Delta_{n-1}}- \frac{\Delta_{n}^*}{\Delta_{n}},
		\end{aligned}
	\end{equation*}
	where $\alpha_0$, $\Delta_{n}$, $\Delta_{n}^*$ are given by
	\begin{equation*}
		\begin{aligned}
			&\alpha_0 = -1,\\
			& \Delta_{n} = \begin{cases}
				& \text{ $\det(s_{i+j-2})_{1\le i \le n, 1\le j \le n}$, \quad if } n\ge 1, \\
				& \text{ $1$,\quad if } n=0, \\
				& \text{ $1$,\quad if } n=-1,
			\end{cases} \\
			& \Delta_{n}^* = \begin{cases}
				& \text{ $\det(s_{i+j-2})_{1\le i \le n, 1\le j \le n+1,j \ne n}$, \quad if } n\ge 1, \\
				& \text{ $0$,\quad if } n=0 .
			\end{cases} 
		\end{aligned}                           
	\end{equation*} 
\end{theorem}

\section{Continued fractions and quadratic orthogonal pairs for hyperelliptic functions} \label{sec:QOP}
In this section, we investigate continued fraction and Laurent series expansions for certain hyperelliptic functions. We also propose the concept of quadratic orthogonal pairs for hyperelliptic functions, for which intimate relationships among their continued fraction expansions  are revealed.
\subsection{Proper quadratic extensions over \texorpdfstring{$\mathbb{C}(X)$}{}}
We first introduce some preliminary results on quadratic extensions. Since one can always complete the square in general, every quadratic extension over $\mathbb{C}(X)$ can be reduced to the following case
\begin{equation*}
	\mathcal{F} = \mathbb{C}(X)[Y]/(Y^2-w),
\end{equation*}
where $w \in \mathbb{C}(X)$ is a rational function in one variable over $\mathbb{C}$.  
\begin{remark}
	When $w$ is a square of an element in $\mathbb{C}(X)$, $\mathcal{F}$ becomes a trivial field extension over $\mathbb{C}(X)$.  
\end{remark}

 On the other hand, due to the fact that every rational function can be expanded into a Laurent series at $\infty$, we have the following embedding of fields
\begin{equation*}
	\begin{aligned}
		\pi_0 :  \mathbb{C}(X) &\longrightarrow  & \mathbb{C}((X^{-1})) \\
		f & \longmapsto & \sum_{n \in \mathbb{Z}} f_n X^n,
	\end{aligned}
\end{equation*}
where $f_n$ are complex numbers given by
\begin{equation*}
	f_n = \frac{1}{2\pi i}\oint_{\infty} X^{-n-1} f dX .
\end{equation*}
The following lemma provides a sufficient and necessary condition to judge whether a given quadratic field extension $\mathcal{F}$ over $\mathbb{C}(X)$ can be embedded into the Laurent series field $\mathbb{C}((X^{-1}))$.

\begin{lemma}
	Suppose $w \in \mathbb{C}(X)$ and $\pi_0(w)=\sum_{n \in \mathbb{Z}} w_n X^n$. Then there exists an embedding $\pi$ from $\mathcal{F}$ to $\mathbb{C}((X^{-1}))$ such that the following diagram commutes
	\[\begin{tikzcd}
		{\mathcal{F}} & {\mathbb{C}((X^{-1}))} \\
		{\mathbb{C}(X)}
		\arrow["\pi", from=1-1, to=1-2]
		\arrow["i", hook, from=2-1, to=1-1]
		\arrow["{\pi_0}"', from=2-1, to=1-2] 
	\end{tikzcd}\]
	if and only if 
	\begin{equation}
		\max(\pi_0(w)) \equiv 0 \quad(mod \quad 2) .\label{proper}
	\end{equation}
Here the definition of $\max$ is given by \eqref{max}.
	\label{proper0}
\end{lemma}
\begin{proof}
	The proof can be achieved by using the Taylor expansion for $\sqrt{1+x}$ in the forward direction and by contradiction in the opposite direction.
\end{proof}

For later use, we call a quadratic field extension over $\mathbb{C}(X)$ a \textit{proper quadratic extension} if and only if it satisfies the condition \eqref{proper} and $w$ is not a square. It turns out that all quadratic extensions induced by hyperelliptic curves \cite{van1,van2,van3,hone2021} are proper.

\subsection{Hyperelliptic functions of genus 1} \label{genus1curve}
\par
Following van der Poorten\cite{van1,van2,van3} and Hone\cite{hone2021}, we consider hyperelliptic curves over $\mathbb{C}$. For our purposes of solving the initial problems of Somos-4 and Somos-5 recurrences, we restrict our attention to the genus 1 case, i.e. the elliptic curve. The constructions in these sections do not actually depend on the explicit formula of the curve and can be easily generalized to an arbitrary proper quadratic field extension over $\mathbb{C}(X)$. 

Consider a curve over $\mathbb{C}$
\begin{equation}
	\mathcal{C}:Y^2 = (X^2+f)^2+4u(X-v) \label{curve}
\end{equation}
where $f$,$u$,$v$ are complex numbers satisfying a certain non-degenerating condition so that the roots of the polynomial on the right hand side are all distinct.
It is evident that the
genus of the curve $\mathcal{C}$ is 1 and it is nothing but an elliptic curve. For convenience, we also
call $\mathcal{C}$ a hyperelliptic curve of genus 1.
As an algebraic variety, the coordinate ring $\mathcal{R}$ of the hyperelliptic curve over $\mathbb{C}$ of genus 1 is given by
\begin{equation*}
	\mathcal{R} = \mathbb{C}[X,Y]/(Y^2 - (X^2+f)^2 - 4u(X-v))
\end{equation*}
and the associated function field is computed by
\begin{equation}
	\mathcal{F} = Frac(\mathcal{R}) \cong  \mathbb{C}(X)[Y]/(Y^2-(X^2+f)^2 - 4u(X-v)) = \{aY+b | a,b \in \mathbb{C}(X) \} .\label{function field}
\end{equation}
Elements in $\mathcal{F}$ are said to be hyperelliptic functions of genus 1.
\subsubsection{Laurent series expansions}
\par
According to \eqref{function field}, the function field $\mathcal{F}$ can be viewed as a two dimensional vector space over $\mathbb{C}(X)$, with $\{1,Y\}$ to be its basis. The only information we know about $Y$ is the algebraic relation given by \eqref{curve}. Actually, there are exactly two such functions satisfying the same algebraic relation \eqref{curve}. These two functions are symmetric under the action of the Galois group $Gal(\mathcal{F}/\mathbb{C}(X))$. Without loss of generality, we simply pick $Y$ to be one of these two functions formally defined by
\begin{equation*}
	Y = \sqrt{(X^2+f)^2+4u(X-v)},
\end{equation*}
while the other such function is given by
\begin{equation*}
	Y^*:= -Y = -\sqrt{(X^2+f)^2+4u(X-v)} . 
	\end{equation*}
\par
Recall that, for any $f \in \mathbb{C}(X)$, we can expand it into a Laurent series when $X \to \infty$, that is, we have the following ring homomorphism
\begin{equation*}
	\begin{aligned}
		\pi_0 :  \mathbb{C}(X) &\longrightarrow  & \mathbb{C}((X^{-1})) \\
		f & \longmapsto & \sum_{n \in \mathbb{Z}} f_n X^n,
	\end{aligned}
\end{equation*}
where $f_n$ are complex numbers given by
\begin{equation*}
	f_n = \frac{1}{2\pi i}\oint_{\infty} X^{-n-1} f dX.
\end{equation*}
Similarly, when $X \to \infty$, since the right hand side of \eqref{curve} satisfies the condition \eqref{proper}, it follows from Lemma \ref{proper0} that $Y$ can also be expanded into a Laurent series
\begin{equation*}
	Y = \sqrt{(X^2+f)^2+4u(X-v)} = \sum_{n \in \mathbb{Z}}y_n X^n  =: (Y)_{\infty},
\end{equation*}
where $y_n$ are complex numbers given by
\begin{equation*}
	y_n = \frac{1}{2\pi i}\oint_{\infty} X^{-n-1} Y dX.
\end{equation*}
Consequently, we can naturally lift $\pi_0$ to $\mathcal{F}$
\[\begin{tikzcd}
	{\mathcal{F}} & {\mathbb{C}((X^{-1}))} \\
	{\mathbb{C}(X)}
	\arrow["\pi", from=1-1, to=1-2]
	\arrow["i", hook, from=2-1, to=1-1]
	\arrow["{\pi_0}"', from=2-1, to=1-2]
\end{tikzcd}\]
where $\pi(f) = \pi(aY+b) := \pi_0(a)(Y)_{\infty} +\pi_0(b)$, for any hyperelliptic function $f = aY+b$.
This means that the ring homomorphism $\pi$ maps a hyperelliptic function to its Laurent series expansion at $\infty$. In the sequel, we also use $(\cdot)_{\infty}$ to denote $\pi$ for convenience.

\subsubsection{Continued fraction expansions}
In the previous subsection, we have shown that the hyperelliptic function field $\mathcal{F}$ can be embedded into the Laurent series field $\mathbb{C}((X^{-1}))$. Recall that,  as indicated in Section \ref{A}, every Laurent series has a continued fraction expansion. Therefore, we conclude that every hyperelliptic function can be expanded into a continued fraction. 

The continued fraction expansions of hyperelliptic functions preserve nice properties. To see this, we start from an arbitrary element $Y_0= r_{0}Y+s_{0}$ in $\mathcal{F}$, where $r_{0}$ and $s_{0}$ are elements in $\mathbb{C}(X)$, $r_{0}$ is nonzero. Under the action of the Galois group $Gal(\mathcal{F}/\mathbb{C}(X))$, we have another hyperelliptic function $Y_{0}^*$ defined by
\begin{equation*}
	Y_{0}^* := r_{0}Y^*+s_{0} = r_{0}(-Y)+s_{0},
\end{equation*}
where $Y_{0}^*$ is said to be the symmetry of $Y_0$.

As is seen in Section \ref{A}, the continued fraction expansion of $Y_0$ is equivalent to the following recursion
\begin{equation}
	Y_n = a_n + \frac{1}{Y_{n+1}} \quad, n \in \mathbb{N},\label{recursion}
\end{equation}
where $a_n := \left \lfloor Y_n \right \rfloor $ for all natural number $n$. Since $r_{0}$ is nonzero and $Y$ is not rational, the above recursion does not terminate. In other words, $Y_{0}$ has an infinite continued fraction expansion. Note that if $Y_n= r_{n}Y+s_{n}$, then $r_{n}$ is also nonzero, otherwise $Y$ would be rational, which is a contradiction.
\par
From a geometric viewpoint, the above recursion can be interpreted as a sequence of projective transformations in the projective space $\mathbb{P}^1(\mathcal{F})$. To be precise, $Y_n$ can be regarded as a point $(Y_n:1)$ in $\mathbb{P}^1(\mathcal{F})$. Thus the above recursion becomes
\begin{equation*}
	\begin{aligned}
		\begin{pmatrix}
			Y_{n+1}\\
			\cdot \cdot\\
			1
		\end{pmatrix} = 
		\begin{pmatrix}
			(Y_{n}-a_n)^{-1}\\
			\cdot \cdot\\
			1
		\end{pmatrix}=
		\begin{pmatrix}
			1\\
			\cdot \cdot\\
			Y_{n}-a_n
		\end{pmatrix} =
		\begin{pmatrix}
			0 & 1 \\
			1 & -a_n
		\end{pmatrix}
		\begin{pmatrix}
			Y_{n}\\
			\cdot \cdot\\
			1
		\end{pmatrix}
	\end{aligned},
\end{equation*}
from which, we obtain
\begin{equation*}
	\begin{pmatrix}
		Y_{n+1}\\
		\cdot \cdot\\
		1
	\end{pmatrix} 
	= \begin{pmatrix}
		0 & 1 \\
		1 & -a_n
	\end{pmatrix}\begin{pmatrix}
		0 & 1 \\
		1 & -a_{n-1}
	\end{pmatrix} \cdots \begin{pmatrix}
		0 & 1 \\
		1 & -a_0
	\end{pmatrix}\begin{pmatrix}
		Y_{0}\\
		\cdot \cdot\\
		1
	\end{pmatrix}.
\end{equation*}
If we introduce
\begin{equation*}
	\begin{aligned}
		M_n := 
		\begin{pmatrix}
			0 & 1 \\
			1 & -a_n
		\end{pmatrix}
	\end{aligned},
\end{equation*} 
for $ n \in \mathbb{N}$,
then recursion \eqref{recursion} is equivalent to the sequence of projective transformations $\{M_n\}_{n \in \mathbb{N}}$. Note that $\{M_n\}_{n \in \mathbb{N}}$ are uniquely determined by $\{a_n\}_{n \in \mathbb{N}}$, thus uniquely determined by $Y_0$.
\par
In general, although $M_n$ transforms $(Y_n:1)$ to $(Y_{n+1}:1)$, it does not certainly transform its symmetries $(Y_n^*:1)$ to $(Y_{n+1}^*:1)$. The following lemma provides a sufficient and necessary condition for that $M_n$ transforms $(Y_n^*:1)$ to $(Y_{n+1}^*:1)$.
\begin{lemma}\label{iff}
	Let $\{Y_n=r_nY+s_n|r_n,s_n\in \mathbb{C}(X),r_n \ne 0 \}_{n\in \mathbb{N}}$ be a sequence in $\mathcal{F}$, and $\{a_n\}_{n \in \mathbb{N}}$ a sequence in $\mathbb{C}(X)$, satisfying
	\begin{equation*}
		Y_n = a_n + \frac{1}{Y_{n+1}}, \quad \forall\, n \in \mathbb{N}.
	\end{equation*}
	Define  
	\begin{equation*}
		\begin{aligned}
			M_n := 
			\begin{pmatrix}
				0 & 1 \\
				1 & -a_n
			\end{pmatrix}.
		\end{aligned}
	\end{equation*}
	Then $M_n$ transforms $(Y_n^*:1)$ to $(Y_{n+1}^*:1)$ if and only if
	\begin{equation}
		r_{n+1} = r_{n}(r_{n+1}^2Y^2-s_{n+1}^2) .\label{iffeq}
	\end{equation} 
\end{lemma} 
\begin{proof}
	The proof of this lemma is straightforward. From the definition of $M_n$, it is clear that $M_n$ transforms $(Y_n^*:1)$ to $(Y_{n+1}^*:1)$ if and only if 
	\begin{equation*}
		Y_{n}^*= a_n+\frac{1}{Y_{n+1}^*},
	\end{equation*}
	which is equivalent to
	\begin{equation}
		r_n(-Y)+s_n= a_n+\frac{1}{r_{n+1}(-Y)+s_{n+1}} .\label{e1}
	\end{equation}
	Recall that  $Y_n$ satisfy
	\begin{equation*}
		Y_{n}= a_n+\frac{1}{Y_{n+1}},
	\end{equation*}
	which can be equivalently written as
	\begin{equation}
		r_nY+s_n= a_n+\frac{1}{r_{n+1}Y+s_{n+1}}. \label{e2}
	\end{equation} 
	It is not hard to see that the sufficient and necessary condition \eqref{iffeq} follows from \eqref{e1} and \eqref{e2}.
\end{proof}
\begin{remark}
	Lemma \ref{iff} does not require $a_n = \left \lfloor Y_n \right \rfloor$. If we have $a_n = \left \lfloor Y_n \right \rfloor$ as an additional condition, then $Y_0$ is exactly expanded into a continued fraction $\langle a_0,\ldots,a_n,\ldots \rangle  $, and the condition \eqref{iffeq} is actually a restriction on the choice of $Y_0$. In fact, Lemma \ref{iff} provides a sufficient and necessary condition to judge whether a given hyperelliptic function is associated with a Lax pair, which will be shown in Theorem \ref{thmlax}. For convenience, we call a hyperelliptic function Lax representable if and only if it satisfies the condition \eqref{iffeq}. It is also noted that $Y_n$ is Lax representable if and only if $Y_n^*$ is Lax representable.
\end{remark}

Since $Y_0 \ne Y_0^*$, $Y_0$ defines a unique projective transformation $\mathcal{L}_0$ as
\begin{equation*}
	\begin{aligned}
		\mathcal{L}_0  :  \mathbb{P}^1(\mathcal{F}) &\longrightarrow   \mathbb{P}^1(\mathcal{F}) \\
		(Y_0:1)& \longmapsto (Y_0:1)  \\
		(Y_0^*:1)& \longmapsto (Y_0^*:1)\,.
	\end{aligned}
\end{equation*}
The projective transformation $\mathcal{L}_0$ can be represented by a linear transformation $L_0$ with eigenvectors $(Y_0,1)$, $(Y_0^*,1)$ and the respective eigenvalues $Y$, $Y^*$. In other words, if we define
\begin{equation*}
	L_0 = \begin{pmatrix}
		A_0&B_0 \\
		C_0&D_0
	\end{pmatrix},
\end{equation*}
then $L_0$ satisfies 
\begin{equation*}
	\begin{pmatrix}
		A_0&B_0 \\
		C_0&D_0
	\end{pmatrix}
	\begin{pmatrix}
		Y_0&Y_0^* \\
		1&1
	\end{pmatrix} = 
	\begin{pmatrix}
		Y_0&Y_0^* \\
		1&1
	\end{pmatrix}
	\begin{pmatrix}
		Y&0 \\
		0&Y^*
	\end{pmatrix}.
\end{equation*}
This uniquely determines $L_0$.  In fact, we have
\begin{equation*}
	\begin{aligned}
		L_0
		= 
		\begin{pmatrix}
			A_0&B_0 \\
			C_0&D_0
		\end{pmatrix}  
		&=  
		\begin{pmatrix}
			Y_0&Y_0^* \\
			1&1
		\end{pmatrix}
		\begin{pmatrix}
			Y&0 \\
			0&Y^*
		\end{pmatrix}
		\begin{pmatrix}
			Y_0&Y_0^* \\
			1&1
		\end{pmatrix}^{-1} 
		\\
		&=\frac{1}{Y_0-Y_0^*}
		\begin{pmatrix}
			Y_0&Y_0^* \\
			1&1
		\end{pmatrix}
		\begin{pmatrix}
			Y&0 \\
			0&Y^*
		\end{pmatrix}
		\begin{pmatrix}
			1&-Y_0^* \\
			-1&Y_0
		\end{pmatrix} \\
		&=
		\begin{pmatrix}
			\frac{YY_0-Y^*Y_0^*}{Y_0-Y_0^*}&-\frac{Y_0Y_{0}^*(Y-Y^*)}{Y_0-Y_0^*} \\
			\frac{Y-Y^*}{Y_0-Y_0^*}&\frac{Y^*Y_0-YY_0^*}{Y_0-Y_0^*}
		\end{pmatrix} \\
		&= 
		\begin{pmatrix}
			\frac{s_0}{r_0}& \frac{r_0^2Y^2-s_0^2}{r_0}\\
			\frac{1}{r_0}& -\frac{s_0}{r_0}
		\end{pmatrix}.
	\end{aligned}
\end{equation*}
The above procedure can still hold when $Y_0$ is replaced by $Y_n=r_nY+s_{n}$, which gives us a sequence of matrices $\{L_n\}_{n \in \mathbb{N}}$ defined by
\begin{equation*}
	L_n = 
	\begin{pmatrix}
		\frac{s_n}{r_n}& \frac{r_n^2Y^2-s_n^2}{r_n}\\
		\frac{1}{r_n}& -\frac{s_n}{r_n}
	\end{pmatrix}
\end{equation*}
satisfying
\begin{equation*}
	L_n
	\begin{pmatrix}
		Y_n&Y_n^* \\
		1&1
	\end{pmatrix} = 
	\begin{pmatrix}
		Y_n&Y_n^* \\
		1&1
	\end{pmatrix}
	\begin{pmatrix}
		Y&0 \\
		0&Y^*
	\end{pmatrix}.
\end{equation*}
It turns out that the sequence $\{L_n\}_{n\in \mathbb{N}}$ is related to the sequence $\{M_n\}_{n \in \mathbb{N}}$, which is illustrated by the following theorem.
\begin{theorem}\label{thmlax}
	Let $\{Y_n=r_nY+s_n|r_n,s_n\in \mathbb{C}(X),r_n \ne 0 \}_{n\in \mathbb{N}}$ be a sequence in $\mathcal{F}$, and $\{a_n\}_{n \in \mathbb{N}}$ a sequence in $\mathbb{C}(X)$, satisfying
	\begin{equation}
		Y_n = a_n + \frac{1}{Y_{n+1}}, \quad \forall\, n \in \mathbb{N} .\label{rec1}
	\end{equation}
	Define  
	\begin{equation*}
		\begin{aligned}
			M_n := 
			\begin{pmatrix}
				0 & 1 \\
				1 & -a_n
			\end{pmatrix},\quad L_n := 
			\begin{pmatrix}
				\frac{s_n}{r_n}& \frac{r_n^2Y^2-s_n^2}{r_n}\\
				\frac{1}{r_n}& -\frac{s_n}{r_n}
			\end{pmatrix}.
		\end{aligned}
	\end{equation*}
	If the condition \eqref{iffeq} holds, that is 
	\begin{equation}
		r_{n+1} = r_{n}(r_{n+1}^2Y^2-s_{n+1}^2),\quad \forall\, n \in \mathbb{N},\label{con1}
	\end{equation}
	then 
	\begin{equation}
		L_{n+1} M_{n} =M_{n} L_{n},\quad \forall\, n \in \mathbb{N},\label{thm}
	\end{equation} 
	and vice versa.
\end{theorem}
\begin{proof}
	The proof is based on the geometric interpretation of $\{L_n\}_{n\in \mathbb{N}}$ and $\{M_n\}_{n \in \mathbb{N}}$. If the condition \eqref{con1} holds, by Lemma \ref{iff}, then we have
	\begin{equation}
		M_n
		\begin{pmatrix}
			Y_n&Y_n^* \\
			1&1
		\end{pmatrix} = 
		\begin{pmatrix}
			Y_{n+1}&Y_{n+1}^* \\
			1&1
		\end{pmatrix}
		\begin{pmatrix}
			g_n&0 \\
			0&h_n
		\end{pmatrix},\label{proofe1}
	\end{equation}
	for some $g_n,h_n \in \mathcal{F}$. Recall that $\{L_n\}_{n\in \mathbb{N}}$ satisfy
	\begin{equation}
		L_n
		\begin{pmatrix}
			Y_n&Y_n^* \\
			1&1
		\end{pmatrix} = 
		\begin{pmatrix}
			Y_n&Y_n^* \\
			1&1
		\end{pmatrix}
		\begin{pmatrix}
			Y&0 \\
			0&Y^*
		\end{pmatrix} .\label{proofe2}
	\end{equation}
	Since $(Y_0,1)$, $(Y_0^*,1)$ forms a basis of $\mathcal{F}^2$. To show \eqref{thm}, it suffices to prove
	\begin{equation*}
		L_{n+1} M_{n} 
		\begin{pmatrix}
			Y_n&Y_n^* \\
			1&1
		\end{pmatrix}=
		M_{n} L_{n}  
		\begin{pmatrix}
			Y_n&Y_n^* \\
			1&1
		\end{pmatrix},
		\quad \forall\, n \in \mathbb{N},
	\end{equation*}
	which immediately follows  from \eqref{proofe1} and \eqref{proofe2}.
	\par
	Conversely, we need to show that \eqref{con1} holds under the assumption that $L_{n+1} M_{n} =M_{n} L_{n}$ for all natural number $n$. By Lemma \ref{iff}, it suffices to show $M_n$ transforms $(Y_n^*:1)$ to $(Y_{n+1}^*:1)$. Note that
	\begin{equation*}
		\begin{aligned}
			L_{n+1} M_{n} 
			\begin{pmatrix}
				Y_n^*\\
				1
			\end{pmatrix} 
			=
			M_{n} L_{n}
			\begin{pmatrix}
				Y_n^*\\
				1
			\end{pmatrix} 
			=Y^*
			M_{n}
			\begin{pmatrix}
				Y_n^*\\
				1
			\end{pmatrix} 
		\end{aligned}.
	\end{equation*}
	This shows $M_{n}
	\begin{pmatrix}
		Y_n^*\\
		1
	\end{pmatrix}$ is an eigenvector of $L_{n+1}$ with the eigenvalue $Y^*$. On the other hand, $\begin{pmatrix}
		Y_{n+1}^*\\
		1
	\end{pmatrix}$ is also an eigenvector of $L_{n+1}$ with the eigenvalue $Y^*$. Thus we have in $\mathbb{P}^1(\mathcal{F})$
	\begin{equation*}
		\begin{aligned}
			M_n
			\begin{pmatrix}
				Y_{n}^*\\
				\cdot \cdot\\
				1
			\end{pmatrix} 
			=
			\begin{pmatrix}
				Y_{n+1}^*\\
				\cdot \cdot\\
				1
			\end{pmatrix}
		\end{aligned},
	\end{equation*}
	which completes the proof.
\end{proof}
\begin{remark}
	Note that Theorem \ref{thmlax} does not require $a_n = \left \lfloor Y_n \right \rfloor$. This means that the equation \eqref{thm} follows once we have two sequences $\{Y_n=r_nY+s_n|r_n,s_n\in \mathbb{C}(X),r_n \ne 0 \}_{n\in \mathbb{N}}$ and $\{a_n\}_{n \in \mathbb{N}}$ satisfying the recursion \eqref{rec1} and the condition \eqref{con1}. In theory of integrable systems, $\{L_n\}_{n\in \mathbb{N}}$ and $\{M_n\}_{n \in \mathbb{N}}$ are known as Lax pairs, and equations \eqref{thm} are often referred to as compatibility conditions. Theorem \ref{thmlax} suggests that the compatibility conditions of Lax pairs $\{L_n\}_{n\in \mathbb{N}}$ and $\{M_n\}_{n \in \mathbb{N}}$ are equivalent to the condition \eqref{con1}, which also implies that
	\begin{equation*}
		\begin{aligned}
			L_{n+1} M_{n} =M_{n} L_{n}  \Leftrightarrow
			\begin{cases}
				& Y_n = a_n + \frac{1}{Y_{n+1}} \\
				& Y^*_n  = a_n + \frac{1}{Y^*_{n+1}}
			\end{cases}.
		\end{aligned}
	\end{equation*}
It is also noted that this setup generalizes the  discussions for the Lax pairs on hyperelliptic curves in \cite{hone2021}, where choice of the corresponding hyperelliptic functions is our special case and satisfy the condition \eqref{con1}. Furthermore, since we do not use the explicit formula of $Y$, the above results can be easily generalized to an arbitrary proper quadratic field extension over $\mathbb{C}(X)$.
\end{remark}

\subsection{Quadratic orthogonal pairs in \texorpdfstring{$\mathcal{F}$}{}}
\label{conventions}
\par
Suppose that we have two hyperelliptic functions $Y_0 = r_0Y+s_0$ and $\Tilde{Y}_0 = \Tilde{r}_0Y+\Tilde{s}_0$, where $r_0$, $s_0$, $\Tilde{r}_0$ and $\Tilde{s}_0$ are all rational functions in $\mathbb{C}(X)$, $r_0$ and $\Tilde{r}_0$ are nonzero. Following the previous section, one can construct $\{Y_n=r_nY+s_n|r_n,s_n\in \mathbb{C}(X),r_n \ne 0 \}_{n\in \mathbb{N}}$, $\{a_n\}_{n \in \mathbb{N}}$, $\{L_n\}_{n\in \mathbb{N}}$ and $\{M_n\}_{n \in \mathbb{N}}$ together with $\{\Tilde{Y}_n=\Tilde{r}_nY+\Tilde{s}_n|\Tilde{r}_n,\Tilde{s}_n\in \mathbb{C}(X),\Tilde{r}_n \ne 0 \}_{n\in \mathbb{N}}$, $\{\Tilde{a}_n\}_{n \in \mathbb{N}}$, $\{\Tilde{L}_n\}_{n\in \mathbb{N}}$ and $\{\Tilde{M}_n\}_{n \in \mathbb{N}}$ satisfying 
\begin{equation*}
	\begin{aligned}
		&a_n = \left \lfloor Y_n \right \rfloor,&&Y_n = a_n +\frac{1}{Y_{n+1}},\\
		&M_n = 
		\begin{pmatrix}
			0 & 1 \\
			1 & -a_n
		\end{pmatrix},
		&&L_n = 
		\begin{pmatrix}
			\frac{s_n}{r_n}& \frac{r_n^2Y^2-s_n^2}{r_n}\\
			\frac{1}{r_n}& -\frac{s_n}{r_n},
		\end{pmatrix},
	\end{aligned}
\end{equation*}
 and 
\begin{equation*}
	\begin{aligned}
		&\Tilde{a}_{n} = \left \lfloor \Tilde{Y}_{n} \right \rfloor,&& \Tilde{Y}_{n} = \Tilde{a}_{n} +\frac{1}{\Tilde{Y}_{n+1}},\\
		&\Tilde{M}_{n} = 
		\begin{pmatrix}
			0 & 1 \\
			1 & -\Tilde{a}_{n}
		\end{pmatrix},
		&&\Tilde{L}_{n} = 
		\begin{pmatrix}
			\frac{\Tilde{s}_{n}}{\Tilde{r}_{n}}& \frac{\Tilde{r}_{n}^2Y^2-\Tilde{s}_{n}^2}{\Tilde{r}_{n}}\\
			\frac{1}{\Tilde{r}_{n}}& -\frac{\Tilde{s}_{n}}{\Tilde{r}_{n}}
		\end{pmatrix},
	\end{aligned}
\end{equation*}
for all natural number $n$. Suppose that $\Tilde{Y}_{0}$ is related  to $Y_0$ according to 
\begin{equation}
	\Tilde{Y}_{0}  Y_0^* = - 1,\label{symmetric}
\end{equation}
then it is natural to ask what are the relations between $\{\{Y_n\}_{n\in \mathbb{N}},\{r_n\}_{n\in \mathbb{N}},\{s_n\}_{n\in \mathbb{N}},\{a_n\}_{n \in \mathbb{N}},$ $ \{L_n\}_{n\in \mathbb{N}},\{M_n\}_{n \in \mathbb{N}}\}$ and $\{\{\Tilde{Y}_n \}_{n\in \mathbb{N}},\{\Tilde{r}_n \}_{n\in \mathbb{N}},\{\Tilde{s}_n \}_{n\in \mathbb{N}}, \{\Tilde{a}_n\}_{n \in \mathbb{N}}, \{\Tilde{L}_n\}_{n\in \mathbb{N}}, \{\Tilde{M}_n\}_{n \in \mathbb{N}}\}$. Here we partially answer this question for a special case that is examined by Hone in \cite{hone2021}. Note that when $\Tilde{Y}_{0}$ and $Y_0^*$ are interpreted as the slopes of two lines in $\mathbb{P}^1(\mathcal{F})$, the relation \eqref{symmetric} is an analogue of the orthogonal condition for these two lines. For this reason, we call a pair of functions $(Y_0,\Tilde{Y}_0)$ a \textit{quadratic orthogonal pair} if the relation \eqref{symmetric} is satisfied.
\par
Now we introduce the choice of $Y_0$ made by Hone in \cite{hone2021}. Consider the hyperelliptic curve of genus 1 given by 
\begin{equation*}
	\mathcal{C}:Y^2 = (X^2+f)^2+4u(X-v)
\end{equation*}
as in Section \ref{genus1curve}. Then one can uniquely define a map $\Phi$ as follows
\begin{equation*}
	\begin{aligned}
		\Phi :  \mathbb{C}^2 &\longrightarrow   \mathbb{C}^2& \\
		(x,y) &  \longmapsto  \left(-x-y^2-f,-\frac{u}{x+y^2+f}-y\right).&
	\end{aligned}
\end{equation*}
It turns out in \cite{hone2021} that the map $\Phi$ is birational, symplectic and integrable with a conserved quantity defined by
\begin{equation}
	H= x^2+fx+ xy^2-uy .\label{H}
\end{equation}
\par
Since $\Phi$ is birational, its inverse $\Phi^{-1}$ is well defined. Therefore if one starts with a point $(d_0,v_0)$ in $\mathbb{C}^2$, one can generate a sequence of points $\{(d_n,v_n)\}_{n \in \mathbb{Z}}$ by setting
\begin{equation}
	(d_n,v_n) := \Phi^n((d_0,v_0)),\quad n \in \mathbb{Z},\label{map}
\end{equation}
provided that $\Phi^n((d_0,v_0)) $ are well defined for all integers $n$.
Note that the well-definedness of $\Phi^n((d_0,v_0))$ is equivalent to the non-zeroness of $d_n$. Under this assumption, if we choose $u_0$ to be another nonzero complex number, we can define a sequence $\{u_n\}_{n \in \mathbb{Z}}$ via $\{d_n\}_{n \in \mathbb{Z}}$ by setting
\begin{equation}
	u_nu_{n-1}+4d_n = 0, \quad n\in \mathbb{Z} .\label{u}
\end{equation}
We conclude that from an initial point $(d_0,v_0,u_0)$ satisfying $d_0u_0\neq0$, we can generate a sequence of points $\{(d_n,v_n,u_n)\}_{n \in \mathbb{Z}}$ satisfying \eqref{map} and \eqref{u} if $\Phi^n((d_0,v_0)) $ are well defined for all integers $n$. Using this sequence, one can choose $\{Y_n=r_n Y+s_n|r_n,s_n\in \mathbb{C}(X),r_n \ne 0 \}_{n\in \mathbb{Z}}$ and $\{a_n\}_{n \in \mathbb{Z}}$ with 
\begin{equation}
	\begin{aligned}
		&r_n := \frac{1}{u_n (X-v_n)},\qquad s_n :=\frac{X^2+f+2d_n}{u_n (X-v_n)},\qquad a_n :=\frac{2}{u_n}(X+v_n).
	\end{aligned}
	\label{choice}
\end{equation}
\par
When given specific conserved quantity $H$, the above choice will satisfy the condition in Theorem \ref{thmlax} together with $a_n = \left \lfloor Y_n \right \rfloor$, which also shows the existence of a Lax pair. This can be summarized in the following lemma.
\begin{lemma}    \label{Ylemma0}
	Suppose $H = -uv$, where $H$ is defined by \eqref{H} with $(x,y)$ being replaced by $(d_n,v_n)$, then the sequences $\{Y_n=r_nY+s_n|r_n,s_n\in \mathbb{C}(X),r_n \ne 0 \}_{n\in \mathbb{Z}}$ and $\{a_n\}_{n \in \mathbb{Z}}$ defined according to \eqref{choice} satisfy the following properties:
	\begin{enumerate}[(i)]
		\item $Y_n= a_n+\frac{1}{Y_{n+1}},\quad \forall\, n \in \mathbb{Z}$;
		\item $a_n = \left \lfloor Y_n \right \rfloor,\quad \forall\, n \in \mathbb{Z}$;
		\item $Y_{n}=\langle a_n,a_{n+1},\ldots \rangle, \quad \forall\, n \in \mathbb{Z} $;
		\item $Y_n$ is J-expressible and Lax representable for any integer $n$.
	\end{enumerate}
\end{lemma}
\begin{proof}
	The proof follows from a special case of Lemma \ref{Ylemma0G} when $g = 1$ in Appendix \ref{generalizepairs}.
\end{proof}
\begin{remark}
The choices of $d_0$ and $v_0$ are not independent since the condition $H= -uv$  is imposed. In addition, under this restriction, it is easy to show that the number of the choices of $(d_0,v_0)$ that do not satisfy the non-degenerating condition is at most countable and they are nothing but the points contained in the two particular orbits that starts or ends at the singularities of the map $\Phi$.
\end{remark}
\par
On the other hand, based on the sequence $\{(d_n,v_n,u_n)\}_{n \in \mathbb{Z}}$ obtained by \eqref{map} and \eqref{u}, one can construct another pair of sequences $\{\Tilde{Y}_n=\Tilde{r}_nY+\Tilde{s}_n|\Tilde{r}_n,\Tilde{s}_n\in \mathbb{C}(X),\Tilde{r}_n \ne 0 \}_{n\in \mathbb{N}}$ and $\{\Tilde{a}_n\}_{n \in \mathbb{N}}$ by setting
\begin{equation}
	\begin{aligned}
		&\Tilde{r}_{n} := \frac{1}{u_{n-1} (X-v_{n-1})}, \quad \Tilde{s}_{n} :=\frac{X^2+f+2d_n}{u_{n-1} (X-v_{n-1})}, \quad \Tilde{a}_{n} :=\frac{2}{u_{n-1}}(X+v_{n-1}).
	\end{aligned}
	\label{choice2}
\end{equation}
The properties of 
$\{\Tilde{Y}_n=\Tilde{r}_nY+\Tilde{s}_n|\Tilde{r}_n,\Tilde{s}_n\in \mathbb{C}(X),\Tilde{r}_n \ne 0 \}_{n\in \mathbb{N}}$ and $\{\Tilde{a}_n\}_{n \in \mathbb{N}}$ are summarized in the following lemma.
\begin{lemma}	\label{Ylemma1}
	Suppose $H = -uv$, where $H$ is defined by \eqref{H} with $(x,y)$ being replaced by $(d_n,v_n)$, then the sequences $\{\Tilde{Y}_n=\Tilde{r}_nY+\Tilde{s}_n|\Tilde{r}_n,\Tilde{s}_n\in \mathbb{C}(X),\Tilde{r}_n \ne 0 \}_{n\in \mathbb{N}}$ and $\{\Tilde{a}_n\}_{n \in \mathbb{N}}$ defined according to \eqref{choice2} satisfy the following properties:
	\begin{enumerate}[(i)]
		\item $\Tilde{Y}_{n}= \Tilde{a}_{n}+\frac{1}{\Tilde{Y}_{n-1}},\quad \forall\, n \in \mathbb{Z}$;
		\item $\Tilde{a}_{n} = \left \lfloor \Tilde{Y}_{n} \right \rfloor,\quad \forall\, n \in \mathbb{Z}$;
		\item $\Tilde{Y}_{n}= \langle \Tilde{a}_{n},\Tilde{a}_{n-1},\ldots \rangle, \quad \forall\, n \in \mathbb{Z}  $;
		\item $\Tilde{Y}_{n}$ is J-expressible and Lax representable for any integer $n$.
	\end{enumerate}
\end{lemma}
\begin{proof}
	The proof follows from a special case of Lemma \ref{Ylemma1G} when $g = 1$ in Appendix \ref{generalizepairs}.
\end{proof}
\par
Moreover, we observe the intimate relationship between two pairs of sequences $\{{a}_n\}_{n \in \mathbb{N}}$,  $\{{Y}_n\}_{n \in \mathbb{N}}$ and  $\{\Tilde{a}_n\}_{n \in \mathbb{N}}$, $\{\Tilde{Y}_n\}_{n \in \mathbb{N}}$, which are given in the following theorem.
\begin{theorem}	\label{pair}
	Suppose that  two pairs of sequences $\{Y_n=r_nY+s_n|r_n,s_n\in \mathbb{C}(X),r_n \ne 0 \}_{n\in \mathbb{Z}}$, $\{a_n\}_{n \in \mathbb{Z}}$ and $\{\Tilde{Y}_n=\Tilde{r}_nY+\Tilde{s}_n|\Tilde{r}_n,\Tilde{s}_n\in \mathbb{C}(X),\Tilde{r}_n \ne 0 \}_{n\in \mathbb{N}}$ and $\{\Tilde{a}_n\}_{n \in \mathbb{N}}$ are defined by \eqref{choice} and \eqref{choice2} respectively, then
	\begin{enumerate}[(i)]
		\item $\Tilde{a}_{n}=a_{n-1},\quad \forall\, n \in \mathbb{Z}$;
		\item  for any integer $n$, $(Y_n,\Tilde{Y}_n)$ is a quadratic orthogonal pair, that is, $\tilde Y_nY_n^*=-1$.
	\end{enumerate}
\end{theorem}
\begin{proof}
	The proof immediately follows  from \eqref{choice} and \eqref{choice2}.
\end{proof}
\begin{remark}

In general, there are analogues of Lemma \ref{Ylemma0}, Lemma \ref{Ylemma1} and Theorem \ref{pair} for general hyperelliptic curves. The general versions of Lemma \ref{Ylemma0}, Lemma \ref{Ylemma1} and Theorem \ref{pair}, along with their proofs, can be found in Appendix \ref{generalizepairs}.
\end{remark}

\subsection{Hankel determinant expressions} \label{sec:Hankelexpressions}

Now we clarify why the notion of a quadratic orthogonal pair is useful in generalizing the Hankel determinant solution of the Somos-4 recurrence to the bilateral case.
\par
Let $n$ be an arbitrary integer. It follows from Lemma 
\ref{Ylemma0}, \ref{Ylemma1} and Theorem \ref{pair} that
\begin{equation*}
	\begin{aligned}
		(Y_{n})^{-1}&=(\langle a_n,a_{n+1},\ldots \rangle  )^{-1} \\
		&=\langle 0,a_n,a_{n+1},\ldots \rangle  
	\end{aligned}
\end{equation*}
and 
\begin{equation*}
	\begin{aligned}
		-(\Tilde{Y}_{n})^{-1} &= -(\langle \Tilde{a}_{n},\Tilde{a}_{n-1},\ldots \rangle  )^{-1} \\
		&=(\langle -\Tilde{a}_{n},-\Tilde{a}_{n-1},\ldots \rangle  )^{-1} \\
		&=\langle 0,-\Tilde{a}_{n},-\Tilde{a}_{n-1},\ldots \rangle   \\
		&=\langle 0,-a_{n-1},-a_{n-2},\ldots \rangle  .
	\end{aligned}
\end{equation*}
Suppose  that $(Y_{n})^{-1}$ and $-(\Tilde{Y}_{n})^{-1}$ are expanded into Laurent series as follows
\begin{equation}
	\begin{aligned}
		(Y_{n})^{-1} &= s_0^{(n)} X^{-1} + s_1^{(n)} X^{-2} + \cdots,\\
		-(\Tilde{Y}_{n})^{-1} & = \Tilde{s}_0^{(n)}X^{-1} + \Tilde{s}_1^{(n)}X^{-2} + \cdots,
	\end{aligned}
\end{equation}
then we have the following equalities
\begin{equation}
	\begin{aligned}
		\langle 0,a_n,a_{n+1},\ldots \rangle   &= s_0^{(n)} X^{-1} + s_1^{(n)} X^{-2} + \cdots, \\
		\langle 0,-a_{n-1},-a_{n-2},\ldots \rangle  & = \Tilde{s}_0^{(n)}X^{-1} + \Tilde{s}_1^{(n)}X^{-2} + \cdots.
	\end{aligned}
\end{equation}
Noticing that $a_k = \frac{2}{u_k}(X+v_k)$, it follows from Theorem \ref{connection} that
\begin{equation*}
	\begin{aligned}
		- (-1)^{-1} \left(\frac{2}{u_n}\right)^{-1}  &= \Delta_{1}^{(n)},\\
		- \left(\frac{2}{u_{n+k-1}}\right)^{-1} \left(\frac{2}{u_{n+k}}\right)^{-1} &= \frac{\Delta_{k-1}^{(n)}\Delta_{k+1}^{(n)}}{(\Delta_{k}^{(n)})^2}, \quad \forall\, k \ge 1,
	\end{aligned}
\end{equation*}
and
\begin{equation*}
	\begin{aligned}
		- (-1)^{-1} \left(-\frac{2}{u_{n-1}}\right)^{-1}  &= \Tilde{\Delta}_{1}^{(n)},\\
		- \left(-\frac{2}{u_{n-k}}\right)^{-1} \left(-\frac{2}{u_{n-k-1}}\right)^{-1} &= \frac{\Tilde{\Delta}_{k-1}^{(n)}\Tilde{\Delta}_{k+1}^{(n)}}{(\Tilde{\Delta}_{k}^{(n)})^2}, \quad \forall\, k \ge 1,
	\end{aligned}
\end{equation*}
where 
\begin{equation*}
	\begin{aligned}
		\Delta_{k}^{(n)} := \det(s^{(n)}_{i+j-2})_{1\le i,j \le k},\quad\qquad \Tilde{\Delta}_{k}^{(n)} := \det(\Tilde{s}^{(n)}_{i+j-2})_{1\le i,j \le k}, \quad \forall\, k \ge 1,
	\end{aligned}
\end{equation*}
and we use the convention $\Delta_0^{(n)} = \Tilde{\Delta}_{0}^{(n)} = 1$. By virtue of \eqref{u}, the above equations become
\begin{equation*}
	\begin{aligned}
		\frac{u_n}{2}  = \Delta_{1}^{(n)}, \qquad d_{n+k} = \frac{\Delta_{k-1}^{(n)}\Delta_{k+1}^{(n)}}{(\Delta_{k}^{(n)})^2}, \quad \forall\, k \ge 1,
	\end{aligned}
\end{equation*}
and
\begin{equation*}
	\begin{aligned}
		-\frac{u_{n-1}}{2}  = \Tilde{\Delta}_{1}^{(n)} \qquad d_{n-k} = \frac{\Tilde{\Delta}_{k-1}^{(n)}\Tilde{\Delta}_{k+1}^{(n)}}{(\Tilde{\Delta}_{k}^{(n)})^2}, \quad \forall\, k \ge 1.
	\end{aligned}
\end{equation*}
If we introduce $\{\tau_{n+k}^{(n)}\}_{k \in \mathbb{Z}}$ by setting
\begin{equation}
	\tau_{n+k}^{(n)} :=
	\left\{
	\begin{aligned}
		&\Delta_{k}^{(n)}, \quad && k \ge 1, \\
		&1, \quad  &&k= 0,\\ 
		&\Tilde{\Delta}_{-k}^{(n)}, \quad && k \le -1,
	\end{aligned}        
	\right. 
	\label{tauseq}
\end{equation}
then $\{d_m\}_{m \in \mathbb{Z}}$ has the following consistent expressions
\begin{equation}
	d_{m} = \frac{\tau^{(n)}_{m-1}\tau^{(n)}_{m+1}}{(\tau^{(n)}_{m})^{2}},\quad \forall\, m \in \mathbb{Z}. \label{d}
\end{equation}
This solves the mismatch problem mentioned in \cite[Remark 4.6]{hone2021}  by Hone, as a result of which, a nicer expression is obtained for the initial problem to the bilateral  Somos-4 recurrence as stated in the following section.  It is noted that, for any integer $n$, the sequence $\{\tau_{m}^{(n)}\}_{m \in \mathbb{Z}}$ defined above by \eqref{tauseq} is uniquely determined by $Y_{n}$. For convenience, we call $\{\tau_{m}^{(n)}\}_{m \in \mathbb{Z}}$ the associated tau sequence with $Y_n$.
\begin{remark}
	It is observed by Hone that the continued fraction expansion of $Y_n$ is associated with a birational symplectic map $\phi_{Y_n}$ on the parameter space \cite[Theorem 3.1]{hone2021}. For a quadratic orthogonal pair $(Y_n,\tilde{Y}_{n})$, the continued fraction expansions of $Y_n$ and $\tilde{Y}_{n}$ actually give a pair of birational symplectic map $(\phi_{Y_n},\phi_{\tilde{Y}_n})$.  In the genus one case, the map $\phi_{Y_{n}}$ sends $(d_n,v_n)$ to $(d_{n+1},v_{n+1})$ while the map $\phi_{\tilde{Y}_{n}}$ sends $(d_n,v_{n-1})$ to $(d_{n-1},v_{n-2})$. By use of the definition of  $\{d_n\}_{n \in \mathbb{Z}}$ and $\{v_n\}_{n \in \mathbb{Z}}$, one can see that these two maps $(\phi_{Y_n},\phi_{\tilde{Y}_{n}})$ are identical. In higher genus cases, the map $\phi_{Y_{n}}$ will send $(\pi_n,\rho_n)$ to $(\pi_{n+1},\rho_{n+1})$ while the map $\phi_{\tilde{Y}_{n}}$ will send $(\pi_n,\rho_{n-1})$ to $(\pi_{n-1},\rho_{n-2})$, where $\pi_n := (\pi_n^{(0)},...,\pi_{n}^{(g-1)})$ and $\rho_n := (\rho_n^{(0)},...,\rho_{n}^{(g-1)})$ are defined in Appendix \ref{generalizepairs}. One can also check that the two maps $(\phi_{Y_n},\phi_{\tilde{Y}_{n}})$ remain identical in higher genus cases by noticing that the iteration for the sequence of polynomials $\{P_n,Q_n\}_{n \in \mathbb{Z}}$ in the positive direction is the same as the iteration for the sequence of polynomials $\{P_n,Q_{n-1}\}_{n \in \mathbb{Z}}$ in the negative direction. It is noted that, in \cite[Section 3]{hone2021}, the rational maps $(\phi, \hat{\phi})$ discussed there are nothing but $(\phi_{Y_n},\phi_{\tilde{Y}_{n}}^{-1})$, which implies that the maps $(\phi, \hat{\phi})$ are not only conjugated but are also inverses to each other, that is, $\phi \circ \hat{\phi} =  \hat{\phi} \circ \phi = \text{id}$.
	\end{remark}

\begin{remark}
	It should also be remarked that, in \cite[Theorem 4.1,4.4]{hone2021}, the two choices of generating functions $(G,G^{\dagger})$ defined there are nothing but $(Y_{1}^{-1},-(\tilde{Y}_0)^{-1})$ in our notations. Therefore the original bilateral sequence obtained by Hone exactly comes from two different rows of the family of tau sequences, that is,  $\{\ldots,\tau^{(0)}_{-2},\tau^{(0)}_{-1},\tau^{(0)}_{0},\tau^{(1)}_{1},\tau^{(1)}_{2},\tau^{(1)}_{3},\ldots\}$, which leads to a ``mismatch" problem. In some sense, the way Hone ``glued" these two parts together is to apply Theorem \ref{rela_n-n-1} to adjust the two parts of the tau sequences coming from two different rows into the same row.
\end{remark}

\section{Applications to bilateral Somos-4 and Somos-5} \label{sec:somos}
The results in the above section can be directly applied to the initial value problems for the bilateral Somos-4 and Somos-5.

\subsection{Connection formulae} \label{sec:connect}

\par
We first present some connection formulae which will be useful in solving the initial value problems for the bilateral Somos-4 and Somos-5.
\par
\begin{lemma}\label{lemmad}
	Suppose $H = -uv$, then the sequence $\{d_n\}_{n \in \mathbb{Z}}$ generated by \eqref{map} satisfies the  recursion
	\begin{equation*}
		d_{n-1}d_{n}^{2}d_{n+1} = \alpha d_{n} + \beta,\quad n\in \mathbb{Z} 
	\end{equation*}
	where $\alpha$ and $\beta$ are defined by 
	\begin{equation}
		\begin{aligned}
			\alpha := u^2,\qquad \beta := u^2(v^2+f).
		\end{aligned}
		\label{coeff}
	\end{equation}
\end{lemma}
\begin{proof}
	This lemma  together with a proof can be found in \cite[Proposition 5.1]{hone2021}. Here we give a slightly different proof which is based our language of being Lax representable. 
	
	For  arbitrary integer $n$, by the definition of $H$, we first have 
	\begin{equation*}
		\begin{aligned}
			-uv = H &= d_n v_n^2-u v_n +d_n^2+ fd_n \\
			&=d_n(v_n^2+d_n+f)-uv_n \\
			&=-d_nd_{n+1}-uv_n.
		\end{aligned}
	\end{equation*}
from which, we have 
	\begin{equation*}
		\begin{aligned}
			d_{n-1}d_{n} = u(v-v_{n-1}), \qquad
			d_{n}d_{n+1} = u(v-v_{n}),
		\end{aligned}
	\end{equation*}
leading to
	\begin{equation}
		\begin{aligned}
			d_{n-1}d_{n}^2d_{n+1} = u^2(v-v_{n-1})(v-v_{n}) .&
		\end{aligned}
		\label{l1}
	\end{equation}
Besides, by Lemma \ref{Ylemma0} and Theorem \ref{thmlax}, since $Y_n$ is Lax representable, the relation \eqref{iffeq} holds. By inserting the expressions in \eqref{choice}, the relation \eqref{iffeq} becomes
	\begin{equation}
		\begin{aligned}
			\frac{1}{u_{n}(X-v_{n})} = \frac{1}{u_{n-1}(X-v_{n-1})} \left(\frac{Y^2}{(u_{n}(X-v_{n}))^2}-\left(\frac{X^2+f+2d_n}{u_{n}(X-v_{n})}\right)^{2}\right),
		\end{aligned}
	\end{equation}
from which  we obtain
	\begin{equation}
		(X^2+f)^2+4u(X-v) = (X^2+f+2d_n)^2+(-4d_n)(X-v_n)(X-v_{n-1}).
	\end{equation}
with the help of \eqref{u}.
	The above equation holds for all $X$ thus holds for a particular evaluation $X= v$, that is,
	\begin{equation}
		(v-v_n)(v-v_{n-1}) = d_n +v^2+f.
		\label{l2}
	\end{equation}
	The proof can be completed by combining \eqref{l1} and \eqref{l2}.
\end{proof}
The above lemma immediately leads to the following corollary.
\begin{coro} \label{corosomos4}
	For any integer $n$, $\{\tau_{m}^{(n)}\}_{m \in \mathbb{Z}}$ defined by \eqref{tauseq} is a bilateral Somos-4 sequence. In other words, we have
	\begin{equation}
		\tau_{m+2}^{(n)} \tau_{m-2}^{(n)} = \alpha\tau_{m+1}^{(n)}\tau_{m-1}^{(n)}+\beta(\tau_{m}^{(n)})^2, \quad \forall\, m,n \in \mathbb{Z},
	\end{equation}
	where $\alpha$ and $\beta$ are defined by \eqref{coeff}.
\end{coro}
This corollary suggests that the sequence of quadratic orthogonal pairs $\{(Y_n,\Tilde{Y}_n)\}_{n \in \mathbb{Z}}$ is identically associated with a sequence of Somos-4 sequences $\{(\ldots,\tau_{-1}^{(n)},\tau_{0}^{(n)},\tau_{1}^{(n)},\ldots)\}_{n \in \mathbb{Z}}$ with the same coefficients $\alpha$ and $\beta$. A further result on this family of tau sequences is given by the following theorem, which shows that all these tau sequences are essentially identical under the gauge transformations.
\begin{theorem}\label{rela_n-n-1}
	The family of tau sequences $\{\tau_{k}^{(n)}\}_{n,k \in \mathbb{Z}}$ defined by (\ref{tauseq}) satisfies the following relation
	\begin{equation}
		\tau_{k}^{(n)} =
		\left\{
		\begin{aligned}
			&(-1)^{k-n+2} \left(\frac{u_{n-1}}{2}\right)^{2n-2k-1}\tau_{k}^{(n-1)}, \quad  &&k\ne n,\\ 
			&1, \quad  &&k= n.\\ 
		\end{aligned}        
		\right. 
	\end{equation}
\end{theorem}
\begin{proof}
	Recall from the definition \eqref{tauseq} that we have in particular
	\begin{equation*}
		\tau_{n +k}^{(n)} =
		\left\{
		\begin{aligned}
			&\Delta_{1}^{(n)} = s_{0}^{(n)} =\frac{u_{n}}{2}, \quad && k = 1, \\
			&1, \quad  &&k= 0,\\ 
			&\Tilde{\Delta}_{1}^{(n)} = \Tilde{s}_{0}^{(n)}= -\frac{u_{n-1}}{2}, \quad && k = -1.
		\end{aligned}        
		\right. 
	\end{equation*}
	By employing the relations  (\ref{u})  and \eqref{d}, we then have, for every integer $n$,
	\begin{equation*}
		(\tau_{n-2}^{(n-1)},\tau_{n-1}^{(n-1)},\tau_{n}^{(n-1)},\tau_{n+1}^{(n-1)}) =\left(-\frac{u_{n-2}}{2},1,\frac{u_{n-1}}{2},-\frac{u_{n-1}^3u_{n}}{16}\right)
	\end{equation*}
	and
	\begin{equation*}
		(\tau_{n-2}^{(n)},\tau_{n-1}^{(n)},\tau_{n}^{(n)},\tau_{n+1}^{(n)}) = \left(-\frac{u_{n-2}u_{n-1}^3}{16},-\frac{u_{n-1}}{2},1,\frac{u_{n}}{2}\right),
	\end{equation*}
which obviously give rise to
	\begin{equation}
		\tau_{k}^{(n)} =(-1)^{k-n+2} \left(\frac{u_{n-1}}{2}\right)^{2n-2k-1}\tau_{k}^{(n-1)} \label{local}
	\end{equation}
	for $k = n-2,n-1,n,n+1$.
On the other hand, from the expression \eqref{d}, we have
    \begin{equation*}
        d_{k} = \frac{\tau_{k-1}^{(n)}\tau_{k+1}^{(n)}}{(\tau_{k}^{(n)})^2} = \frac{\tau_{k-1}^{(n-1)}\tau_{k+1}^{(n-1)}}{(\tau_{k}^{(n-1)})^2},\quad \forall\, k \in \mathbb{Z}.
    \end{equation*}
    Then it follows by induction on $k$ that the relation \eqref{local} holds automatically for all $k \in \mathbb{Z}$. 

\end{proof}

Based on the above theorem, we readily obtain the following corollary.
\begin{coro}\label{coro_tau} 
	For $n,k\in \mathbb{Z}$,  any  $\tau_{k}^{(n)}$ can be represented as a monomial in $\{u_m\}_{m \in \mathbb{Z}}$
	\begin{equation}
		\tau_{k}^{(n)} =
		\left\{
		\begin{aligned}
			&\prod_{m=k}^{n-1}(-1)^{k-m+1}\left(\frac{u_m}{2}\right)^{2m-2k+1}, \quad  &n > k,\\ 
			&1, \quad  &n= k,\\ 
			&\prod_{m=n}^{k-1}(-1)^{k-m+1}\left(\frac{u_m}{2}\right)^{2k-2m-1}, \quad  &n < k.\\
		\end{aligned}        
		\right. 
	\end{equation}
\end{coro}
\begin{remark}
It is noted that the explicit formula for the recursion of the sequence $\{d_{n}\}_{n \in \mathbb{Z}}$  in Lemma \ref{lemmad} and the bilinear recurrence relation for the associated tau sequence in Corollary \ref{corosomos4} are only valid for the case of genus 1. However,
Theorem \ref{rela_n-n-1} and Corollary \ref{coro_tau} remain valid for hyperelliptic curves of any genus, since the same relations \eqref{u} and \eqref{d} hold for the general setting for hyperelliptic curves of any genus.
\end{remark}

\subsection{Initial value problem for Somos-4}
The general initial value problem for the bilateral Somos-4 recurrence can be formulated as
\begin{equation}
	\begin{cases}
		&  \tau^*_{n+2}\tau^*_{n-2} =\alpha \tau^*_{n+1}\tau^*_{n-1} + \beta (\tau^*_{n})^2, \, n\in \mathbb{Z},\\
		&  \tau^*_{-1} =A, \quad  \tau^*_{0} =B, \quad \tau^*_{1} =C, \quad \tau^*_{2} =D,
	\end{cases}
	\label{initial40}
\end{equation}
where $\alpha$, $\beta$ and $\{A,B,C,D\}$ are indeterminants. Since the above recurrence is invariant under gauge transformations $\tau^*_{n} \mapsto ab^n\tau^*_{n}$ for all complex numbers $a$ and $b$, the problem \eqref{initial40} is equivalent to the problem

\begin{equation}
	\begin{cases}
		&  \tau_{n+2}\tau_{n-2} =\alpha \tau_{n+1}\tau_{n-1} + \beta (\tau_{n})^2, \, n\in \mathbb{Z},\\
		&  \tau_{-1} =x, \quad  \tau_{0} =1, \quad \tau_{1} =y, \quad \tau_{2} =z .
	\end{cases}
	\label{initial4}
\end{equation}
We are going to solve the problem \eqref{initial4}.

Based on the argument in Subsection \ref{conventions} and \ref{sec:connect}, in order to solve \eqref{initial4}, it suffices to make a choice of the parameters $\{d_0,v_0,u_0,u,v,f\}$ such that the associated tau sequence of the hyperelliptic function $Y_0$ is a solution to problem \eqref{initial4}. Recall that $Y_0$ is defined by
\begin{equation}
	\begin{aligned}
		Y_0 =r_0Y+s_0= \frac{\sqrt{(X^2+f)^2+4u(X-v)}+X^2+f+2d_0}{u_0(X-v_0)}.\\
	\end{aligned}
	\label{Y00}
\end{equation}
The following theorem shows that, for each triple of initial values $(x,y,z)$ leading to a unique solution to the problem \eqref{initial4}, there exist exactly two choices of the parameters $\{d_0,v_0,u_0,u,v,f\}$ such that the solution of problem \eqref{initial4} is given by the associated tau sequence of the hyperelliptic function $Y_0$.
\begin{theorem}
	Let $\{\tau_{m}^{(0)}\}_{m \in \mathbb{Z}}$ be the associated tau sequence defined by \eqref{tauseq} for the hyperelliptic function $Y_0$ \eqref{Y00}. Then $\{\tau_{m}^{(0)}\}_{m \in \mathbb{Z}}$ is a solution to the problem \eqref{initial4} on the bilateral  Somos-4 if and only if the corresponding parameters $\{d_0,v_0,u_0,u,v,f\}$ are given by
	\begin{equation}
		\begin{cases}
			&  d_0 = xy,\\
			&  v_0 = \frac{\beta y^2-x^2z^2+\alpha z + \alpha x y^3}{2\sqrt{\alpha}xyz},\\
			& u_0=2y,\\
			& u = \sqrt{\alpha},\\
			& v = \frac{\beta y^2+x^2z^2+\alpha z + \alpha x y^3}{2\sqrt{\alpha}xyz},\\
			& f = \frac{\beta}{\alpha} - \frac{(\beta y^2+x^2z^2+\alpha z + \alpha x y^3)^2}{4\alpha x^2y^2z^2},
		\end{cases}
		\label{c1}
	\end{equation}
	or
	\begin{equation}
		\begin{cases}
			&  d_0 = xy,\\
			&  v_0 = -\frac{\beta y^2-x^2z^2+\alpha z + \alpha x y^3}{2\sqrt{\alpha}xyz},\\
			& u_0=2y,\\
			& u = -\sqrt{\alpha},\\
			& v = -\frac{\beta y^2+x^2z^2+\alpha z + \alpha x y^3}{2\sqrt{\alpha}xyz},\\
			& f = \frac{\beta}{\alpha} - \frac{(\beta y^2+x^2z^2+\alpha z + \alpha x y^3)^2}{4\alpha x^2y^2z^2}.
		\end{cases}
		\label{c2}
	\end{equation}
\end{theorem}
\begin{proof}
{Suppose that the sequence $\{\tau_{m}^{(0)}\}_{m \in \mathbb{Z}}$ shares the same initial data as problem \eqref{initial4}, i.e.
\begin{equation*}
    \tau^{(0)}_{-1} = x, \tau^{(0)}_{0} = 1, \tau^{(0)}_{1} = y, \tau^{(0)}_{2}= z.
\end{equation*}}
Then the formulae for $d_0$ and $u_0$ can be readily obtained from the initial data, that is, 
\begin{equation*}
		d_0 = \frac{\tau^{(0)}_{-1}\tau^{(0)}_{1}}{(\tau^{(0)}_{0})^2} =xy, \quad
		d_1 = \frac{\tau^{(0)}_{0}\tau^{(0)}_{2}}{(\tau^{(0)}_{1})^2}= \frac{z}{y^2},
		\quad
		\frac{u_0}{2} = \tau^{(0)}_1= y.
\end{equation*}
It follows from Lemma \ref{lemmad} that
\begin{equation}
	u^2	 = \alpha,\quad
	u^2(v^2+f)	 = \beta, \label{uvfequation}
\end{equation}
from which we obtain
\begin{equation} 
\label{ffo}
		f = \frac{\beta}{u^2} - u^2 \left(\frac{v}{u}\right)^2 = \frac{\beta}{\alpha} - \alpha \left(\frac{v}{u}\right)^2.
\end{equation}
Recall that, for the birational map $\Phi$ defined by \eqref{map}, we have 
\begin{equation}
	\begin{aligned}\label{vn}
		H = -uv = d_n(d_n+ v_n^2 + f) -uv_n= -d_nd_{n+1}-uv_n
	\end{aligned}
\end{equation}
and
\begin{equation} \label{iterationpro}
	d_{n+1} +d_{n} +v_{n}^2 +f = 0.
\end{equation}
By eliminating $v_n$ from \eqref{vn} and \eqref{iterationpro}  and employing the formula \eqref{uvfequation} we get
\begin{equation} 
	\begin{aligned} \label{vdu}
		\frac{v}{u} = \frac{1}{2}\left(\frac{\beta}{\alpha}\frac{1}{d_0d_1}+\frac{d_0d_1}{\alpha}+\frac{1}{d_0}+\frac{1}{d_1}\right) 
		 = \frac{\beta y^2 +x^2z^2+\alpha z +\alpha xy^3}{2\alpha xyz}.
	\end{aligned}
\end{equation}
The formula for the coefficient $f$ with respect to the initial data $\{x,y{,z},\alpha,\beta \}$ can be obtained by substituting \eqref{vdu} into \eqref{ffo}. Furthermore, it is noted that $v_0$ can be obtained from \eqref{vn}
\begin{equation*}
	v_0 = v - \frac{xz}{uy}.
\end{equation*}
Based on the above relations, it is not hard to see that all the parameters are determined once the sign of $u=\pm\sqrt\alpha$ is fixed. This implies that the both cases will solve the initial problem and therefore the conclusion follows.
\end{proof}

\begin{remark}
	It is noted that the two choices of parameters give us exactly two choices of hyperelliptic functions. If we use $Y_{0}$ and $Y_{0}’$ to denote them, from the relations between the two choices of parameters, one can easily deduce that $$Y_{0}(X) + Y_{0}'(-X) = 0.$$ This will lead to $$(Y_{0})^{-1}(X)+ (Y_{0}')^{-1}(-X) = 0$$ and $$(-\Tilde{Y}_{0})^{-1}(X)+ (-\Tilde{Y}'_0)^{-1}(-X) = 0,$$ which means that, if the generating function of one choice is $s_0X^{-1}+s_1X^{-2}+s_{2}X^{-3}+\cdots$, then the generating function of the other choice will be $s_0X^{-1}-s_1X^{-2}+s_{2}X^{-3}-\cdots$. It is remarkable that the two sequences $\{s_0,s_1,s_2,s_3,\cdots\}$ and $\{s_0,-s_1,s_2,-s_3,\cdots\}$ will lead to the same Hankel determinant sequence, namely, $\det(s_{i+j-2})_{1\le i,j \le n} = \det((-1)^{i+j}s_{i+j-2})_{1\le i,j \le n}$ for any $n \in \mathbb{N}^*$. This fact can also be seen directly from the property of determinants, since we have in general $\det(a_{i,j}) = \det((-1)^{i+j}a_{i,j})$. 
	\end{remark}

The above theorem implies that the solution to  the bilateral  Somos-4 recurrence can be expressed by Hankel determinants. To be precise, suppose that we have
\begin{equation*}
	\begin{aligned}
		&G:= (Y_0)^{-1} = \frac{\sqrt{(X^2+f)^2+4u(X-v)} -X^2-f-2d_0}{u_{-1}(X-v_{-1})}=   s_0X^{-1}+s_{1}X^{-2}+\cdots,\\
		&G^{\dagger}:= (-\Tilde{Y}_0)^{-1} =  \frac{-\sqrt{(X^2+f)^2+4u(X-v)} +X^2+f+2d_0}{u_{0}(X-v_{0})}=s^{\dagger}_0X^{-1}+s^{\dagger}_{1}X^{-2}+\cdots,
	\end{aligned}
\end{equation*}
with $u_{-1}$ and $v_{-1}$ given by
\begin{equation}
	\begin{aligned}
		u_{-1} = -2x,\qquad v_{-1} = \frac{u}{d_0} -v_0 = \frac{u(-\beta y^2 +x^2z^2+\alpha z - \alpha x y^3)}{2\alpha xyz}.
	\end{aligned}
	\label{c-1}
\end{equation}
Furthermore, it is not hard to see that $G$ and $G^{\dagger}$ satisfy the following algebraic relations
\begin{equation*}
	\begin{aligned}
		u_{-1}(X-v_{-1}) G^2 + 2 (X^2+f+2d_0) G &= u_0(X-v_0), \\
		u_0(X-v_0)(G^{\dagger})^2 - 2 (X^2+f+2d_0) G^{\dagger} &= u_{-1}(X-v_{-1}),
	\end{aligned}
\end{equation*}
from which, one can uniquely determine $\{s_n\}_{n \in \mathbb{N}}$ and $\{s^{\dagger}_n\}_{n \in \mathbb{N}}$ by recursion. In fact, we obtain
\begin{equation}
	\begin{aligned}
		s_0 &= \frac{u_0}{2},\\
		s_1 &= -\frac{1}{2}u_0v_0,\\
		s_2 &=-(f+2d_0)s_0- \frac{1}{2}u_{-1}\left(\sum_{i+j= 0}s_is_j\right), \\
		s_k &= -(f+2d_0)s_{k-2}  -\frac{1}{2}u_{-1}\left(\sum_{i+j= k-2}s_is_j\right)+\frac{1}{2}u_{-1}v_{-1}\left(\sum_{i+j= k-3}s_is_j\right),\quad k\ge 3,
	\end{aligned}
	\label{rec_s1}
\end{equation}
and
\begin{equation}
	\begin{aligned}
		s^{\dagger}_0 &= -\frac{u_{-1}}{2}, \\
		s^{\dagger}_1 &= \frac{1}{2}u_{-1}v_{-1},\\
		s^{\dagger}_2 &=-(f+2d_0)s^{\dagger}_0+ \frac{1}{2}u_{0}\left(\sum_{i+j= 0}s^{\dagger}_is^{\dagger}_j\right),\\
		s^{\dagger}_k &= -(f+2d_0)s^{\dagger}_{k-2}  +\frac{1}{2}u_{0}\left(\sum_{i+j= k-2}s^{\dagger}_is^{\dagger}_j\right)-\frac{1}{2}u_{0}v_{0}\left(\sum_{i+j= k-3}s^{\dagger}_is^{\dagger}_j\right),\quad k\ge 3.
	\end{aligned}
	\label{rec_s2}
\end{equation}

Eventually, we have the following theorem.
\begin{theorem}
The solution $\{\tau_n\}_{n \in \mathbb{Z}}$ to the problem \eqref{initial4} on the  bilateral  Somos-4 recurrence can be given by the Hankel determinants with $\{s_k,s^{\dagger}_k\}_{k\in \mathbb{N}}$ in \eqref{rec_s1} and \eqref{rec_s2} as elements, that is,
\begin{equation*}
	\tau_{n} =
	\left\{
	\begin{aligned}
		&\det(s_{i+j-2})_{1 \le i,j\le  n}, \quad && n \ge 1, \\
		&1, \quad  &&n= 0,\\ 
		&\det(s^{\dagger}_{i+j-2})_{1 \le i,j\le  -n}, \quad && n \le -1. \\
	\end{aligned}    
	\right.
\end{equation*}
\end{theorem}
\begin{remark}
It is noted that the Laurent phenomenon can follow from the Hankel determinant solution. On one hand,  we have from \eqref{c1},\eqref{c2} and \eqref{c-1}
\begin{equation*}
	\{d_0,u_0,v_0,u_{-1},v_{-1},f\} \subseteq \mathbb{Z}\left[{\frac{1}{2},}\frac{1}{\sqrt{\alpha}},\alpha,\beta,x^{\pm 1},y^{\pm 1},{z^{\pm 1}}\right].
\end{equation*}
Since $\mathbb{Z}\left[{\frac{1}{2},}\frac{1}{\sqrt{\alpha}},\alpha,\beta,x^{\pm 1},y^{\pm 1} {,z^{\pm 1}}\right]$ is closed under addition and multiplication, we conclude that
\begin{equation}
	\tau_n \in  \mathbb{Z}\left[{\frac{1}{2},}\frac{1}{\sqrt{\alpha}},\alpha,\beta,x^{\pm 1},y^{\pm 1}{,z^{\pm 1}}\right],\quad \forall\, n \in \mathbb{Z}.
	\label{lauren1}
\end{equation}
On the other hand, since $\{\tau_n\}_{n\in \mathbb{Z}}$ satisfies the recursion \eqref{initial4}, any $\tau_n$ is a rational function in $\{\alpha,\beta,x,y{,z}\}$, that is 
\begin{equation}
	\tau_n \in  \mathbb{Q}(\alpha,\beta,x,y{,z}),\quad \forall\, n \in \mathbb{Z}.
	\label{laurent2}
\end{equation}
Combining \eqref{lauren1} and \eqref{laurent2}, we obtain
\begin{equation*}
	\tau_n \in  \mathbb{Z}\left[{\frac{1}{2},}\frac{1}{\sqrt{\alpha}},\alpha,\beta,x^{\pm 1},y^{\pm 1}{,z^{\pm 1}}\right] \cap \mathbb{Q}(\alpha,\beta,x,y{,z}) \subseteq \mathbb{Z}\left[{\frac{1}{2},}\alpha^{\pm 1},\beta,x^{\pm 1},y^{\pm 1}{,z^{\pm 1}}\right],\quad \forall\, n \in \mathbb{Z}.
\end{equation*}
In addition, observe that the recurrence \eqref{initial4} is also valid for $\alpha = 0$, which means that  
\begin{equation*}
	\tau_n \in  \mathbb{Z}\left[{\frac{1}{2},}\alpha,\beta,x^{\pm 1},y^{\pm 1}{,z^{\pm 1}}\right],\quad \forall\, n \in \mathbb{Z}.
\end{equation*}
We claim that the half integer $\frac{1}{2}$ can be removed using an argument by contradiction. Assume that $n_0$ is the first
positive integer such that
$$
\tau_{n_0} \in  \mathbb{Z}\left[\frac{1}{2},\alpha,\beta,x^{\pm 1},y^{\pm 1}{,z^{\pm 1}}\right]
$$
but
$$
\tau_{n_0} \notin  \mathbb{Z}\left[\alpha,\beta,x^{\pm 1},y^{\pm 1}{,z^{\pm 1}}\right].
$$
By considering index shift in the Somos-4 recurrence and applying the contradiction hypothesis, it is not hard to deduce that
$$
\tau_{n_0} \in  \mathbb{Z}\left[\alpha,\beta,y^{\pm 1}{,z^{\pm 1}},\left(\frac{\alpha z+\beta y^2}{x}\right)^{\pm 1}\right].
$$
A further analysis will lead to 
$$
\tau_{n_0} \in  \mathbb{Z}\left[\alpha,\beta,x^{\pm 1},y^{\pm 1}{,z^{\pm 1}}\right],
$$
which contradicts to the choice of $n_0$. Therefore we conclude that such positive $n_0$ does not exist. A similar argument also implies that there does not exist such negative $n_0$, thus leading to
$$	
\tau_n \in  \mathbb{Z}\left[\alpha,\beta,x^{\pm 1},y^{\pm 1}{,z^{\pm 1}}\right],\quad \forall\, n \in \mathbb{Z}.
$$
Finally, 
based on the gauge invariance, one can conclude that the bilateral Somos-4 \eqref{initial40} exhibits the Laurent phenomenon.
\end{remark}

\begin{remark}
It is noted that the Hankel determinant solution to half of the Somos-4 with $\tau_{-1}=\tau_0=1,\tau_{1}=x,\tau_2=y$ has been derived in  \cite{chang2015hankel}, while here our focus is the Hankel determinant solution to the bilateral Somos-4 with $\tau_{-1}=x, \tau_0=1,\tau_{1}=y,\tau_2=z$, for which the reduction formula to the half case is formally different from that in  \cite{chang2015hankel}. 
\end{remark}

\subsection{Initial value problem for Somos-5}
Now we consider the general initial value problem for the bilateral Somos-5 recurrence
\begin{equation}
	\begin{cases}
		&  \tau^*_{n+5}\tau^*_{n} =\alpha \tau^*_{n+4}\tau^*_{n+1} + \beta \tau^*_{n+3}\tau^*_{n+2},\qquad n \in \mathbb{Z},\\
		&  \tau^*_{0} =A, \quad  \tau^*_{1} =B, \quad \tau^*_{2} =C, \quad \tau^*_{3} =D,\quad  \tau^*_{4} =E,
	\end{cases}
	\label{initial50}
\end{equation}
where $\alpha$, $\beta$ and $\{A,B,C,D,E\}$ are indeterminants. Since the above recursion equation is invariant under gauge transformations $\tau^*_{n} \mapsto ab^n\tau^*_{n}$ for all complex numbers $a$ and $b$, the problem \eqref{initial50} can be equivalently transformed into
\begin{equation}
	\begin{cases}
		&  \tau_{n+5}\tau_{n} =\alpha \tau_{n+4}\tau_{n+1} + \beta \tau_{n+3}\tau_{n+2},\\
		&  \tau_{0} = 1, \quad  \tau_{1} = 1, \quad \tau_{2} =x, \quad \tau_{3} =y,\quad \tau_{4} =z.
	\end{cases}
	\label{initial5}
\end{equation}

In order to solve this problem, we recall a connection between Somos-4 and Somos-5, which was originally shown in  \cite[Proposition 2.8]{hone2007sigma} and then interpreted as a B\"acklund transformation in \cite[Section 3.2]{chang2015hankel}. We claim that this relationship  applies to not only the half case, but also the bilateral case and it enables us to solve the above problem for Somos-5.

\begin{lemma}
	Suppose that $\{\tau_n\}_{n \in \mathbb{Z}}$ satisfies 
	\begin{equation*}
		\tau_{n+5}\tau_{n} =\alpha \tau_{n+4}\tau_{n+1} + \beta \tau_{n+3}\tau_{n+2},\quad \forall\, n \in \mathbb{Z}.
	\end{equation*}
	If we define
	\begin{equation*}
		\Delta^{(1)}_{n} := \tau_{2n},\qquad \Delta^{(2)}_{n} := \tau_{2n+1},\quad \forall\, n \in \mathbb{Z}, 
	\end{equation*}
	then $\{\Delta^{(1)}_n\}_{n \in \mathbb{Z}}$ and $\{\Delta^{(2)}_n\}_{n \in \mathbb{Z}}$ satisfy the same Somos-4 recurrence, that is, 
	\begin{equation*}
		\Delta^{(1)}_{n+4}\Delta^{(1)}_{n}=
		\Tilde{\alpha}
		\Delta^{(1)}_{n+3}\Delta^{(1)}_{n+1}+\Tilde{\beta}\left(\Delta^{(1)}_{n+2}\right)^2,\quad \forall\, n \in \mathbb{Z} 
	\end{equation*}
	and 
	\begin{equation*}
		\Delta^{(2)}_{n+4}\Delta^{(2)}_{n}=
		\Tilde{\alpha}
		\Delta^{(2)}_{n+3}\Delta^{(2)}_{n+1}+\Tilde{\beta}\left(\Delta^{(2)}_{n+2}\right)^2,\quad \forall\, n \in \mathbb{Z},
	\end{equation*}
	where $\Tilde{\alpha}$ and $\Tilde{\beta}$ are defined by 
	\begin{equation*}
		\begin{aligned}
			&\Tilde{\alpha} = \beta^2,\qquad \Tilde{\beta} = \alpha(2\beta^2+\alpha\beta J+\alpha^3),
		\end{aligned}
	\end{equation*}
	and $J$ is a conserved quantity computed by 
	\begin{equation*}
		\begin{aligned}
			&J = h_{n-1}+h_{n} +\alpha\left(\frac{1}{h_{n-1}}+\frac{1}{h_{n}}\right)+\frac{\beta}{h_{n-1}h_{n}},\qquad
			h_n =\frac{\tau_{n-1}\tau_{n+2}}{\tau_{n}\tau_{n+1}} .
		\end{aligned}
	\end{equation*}
\end{lemma}
By employing the above lemma, the problem \eqref{initial5} can be reformulated as the solution of the following two problems:
\begin{equation}
	\begin{cases}
		&  \Delta^{(1)}_{n+2}\Delta^{(1)}_{n-2} =\Tilde{\alpha} \Delta^{(1)}_{n+1}\Delta^{(1)}_{n-1} + \Tilde{\beta} (\Delta^{(1)}_{n})^2,\\
		&  \Delta^{(1)}_{-1} =\frac{\alpha \beta x^2+\alpha^2 xy +\beta z }{yz},\\ 
		& \Delta^{(1)}_{0} =1,\\ 
		& \Delta^{(1)}_{1} =x,\\
		& \Delta^{(1)}_{2} =z,
	\end{cases}
	\label{delta1}
\end{equation}
and
\begin{equation}
	\begin{cases}
		&  \Delta^{(2)}_{n+2}\Delta^{(2)}_{n-2} =\Tilde{\alpha}\Delta^{(2)}_{n+1}\Delta^{(2)}_{n-1} + \Tilde{\beta}(\Delta^{(2)}_{n})^2,\\
		&  \Delta^{(2)}_{-1} =\frac{\beta x +\alpha y}{z},\\ 
		&  \Delta^{(2)}_{0} =1,\\
		&\Delta^{(2)}_{1} =y,\\
		&\Delta^{(2)}_{2} =\alpha z+ \beta xy,
	\end{cases}
	\label{delta2}
\end{equation}
where $\Tilde{\alpha}$ and $\Tilde{\beta}$ are given by
\begin{equation}
	\begin{aligned}
		\Tilde{\alpha} = \beta^2, \qquad
		\Tilde{\beta} = 2\alpha\beta^2+\alpha^2\beta \left(\frac{y^2z+z^2+\beta x^3y+\alpha x^2 z+\alpha x^2y^2}{xyz}\right)+\alpha^4 .
	\end{aligned}
\end{equation}
These two problems are nothing but two specific cases of the problem \eqref{initial40} which has been solved in the previous subsection. As a result, one can derive the solutions to \eqref{delta1} and  \eqref{delta2}, and consequently solve the problem \eqref{initial50}.
\par
First, upon comparing the problems \eqref{delta1} and \eqref{delta2} with the problem \eqref{initial40} together with the equivalent formulation \eqref{initial4}, we can set
\begin{equation*}
	\begin{aligned}
		x^{(1)} &= \frac{\alpha \beta x^2+\alpha^2 xy +\beta z }{yz},\quad& y^{(1)} &=x,\quad &z^{(1)} &=z,\\
		x^{(2)} &=\frac{\beta x+\alpha y}{z}, \quad&y^{(2)} &= y,\quad& z^{(2)} &=\alpha z + \beta xy,
	\end{aligned}
\end{equation*}
and introduce the parameters
\begin{equation*}
	\begin{cases}
		&  d_0^{(r)} = x^{(r)}y^{(r)},\\
		&  v_0^{(r)} = \frac{\Tilde{\beta} (y^{(r)})^2-(x^{(r)})^2(z^{(r)})^2+\Tilde{\alpha} z^{(r)} + \Tilde{\alpha} x^{(r)} (y^{(r)})^3}{2\sqrt{\Tilde{\alpha}}x^{(r)}y^{(r)}z^{(r)}},\\
		& u_0^{(r)}=2y^{(r)},\\
		& u^{(r)} = \sqrt{\Tilde{\alpha}},\\
		& v^{(r)} = \frac{\Tilde{\beta} (y^{(r)})^2+(x^{(r)})^2(z^{(r)})^2+\Tilde{\alpha} z^{(r)} + \Tilde{\alpha} x^{(r)} (y^{(r)})^3}{2\sqrt{\Tilde{\alpha}}x^{(r)}y^{(r)}z^{(r)}},\\
		& f^{(r)} = \frac{\Tilde{\beta}}{\Tilde{\alpha}} - \frac{(\Tilde{\beta} (y^{(r)})^2+(x^{(r)})^2(z^{(r)})^2+\Tilde{\alpha} z^{(r)} + \Tilde{\alpha} x^{(r)} (y^{(r)})^3)^2}{4\Tilde{\alpha} (x^{(r)})^2(y^{(r)})^2(z^{(r)})^2},\\
		&u_{-1}^{(r)} = -2x^{(r)},\\
		&v_{-1}^{(r)} = \frac{u^{(r)}}{d_0^{(r)}}-v_0^{(r)},
	\end{cases}
\end{equation*}
for $r = 1,2$. Then we define two pairs of hyperelliptic functions with Laurent series expansions as follows
\begin{equation*}
	\begin{aligned}
		&G_{r}:=\frac{\sqrt{(X^2+f^{(r)})^2+4u^{(r)}(X-v^{(r)})} -X^2-f^{(r)}-2d_0^{(r)}}{u_{-1}^{(r)}(X-v_{-1}^{(r)})}=   s_{0,r}X^{-1}+s_{1,r}X^{-2}+\cdots,\\
		&G_r^{\dagger}:=\frac{-\sqrt{(X^2+f^{(r)})^2+4u^{(r)}(X-v^{(r)})} +X^2+f^{(r)}+
			2d_0^{(r)}}{u_{0}^{(r)}(X-v_{0}^{(r)})}=s^{\dagger}_{0,r}X^{-1}+s^{\dagger}_{1,r}X^{-2}+\cdots.
	\end{aligned}
\end{equation*}
Based on the intrinsic algebraic relations, we can show that $\{s_{n,r},s_{n,r}^{\dagger}\}_{n \in \mathbb{N},r = 1,2}$ satisfy the following recursion relations
\begin{equation}
	\begin{aligned}
		s_{0,r} &= \frac{u_0^{(r)}}{2},\\
		s_{1,r} &= -\frac{1}{2}u_0^{(r)}v_0^{(r)},\\
		s_{2,r} &=-(f^{(r)}+2d_0^{(r)})s_{0,r}- \frac{1}{2}u_{-1}^{(r)}\left(\sum_{i+j= 0}s_{i,r}s_{j,r}\right), \\
		s_{k,r} &= -(f^{(r)}+2d_0^{(r)})s_{k-2,r}  -\frac{1}{2}u_{-1}^{(r)}\left(\sum_{i+j= k-2}s_{i,r}s_{j,r}\right)+\frac{1}{2}u_{-1}^{(r)}v_{-1}^{(r)}\left(\sum_{i+j= k-3}s_{i,r}s_{j,r}\right),\quad k\ge 3,
	\end{aligned}
	\label{rec-nodagger}
\end{equation}
\begin{equation}
	\begin{aligned}
		s^{\dagger}_{0,r} &= -\frac{u_{-1}^{(r)}}{2},\\
		s^{\dagger}_{1,r} &= \frac{1}{2}u_{-1}^{(r)}v_{-1}^{(r)},\\
		s^{\dagger}_{2,r} &=-(f^{(r)}+2d_0^{(r)})s^{\dagger}_{0,r}+ \frac{1}{2}u_{0}^{(r)}\left(\sum_{i+j= 0}s^{\dagger}_{i,r}s^{\dagger}_{j,r}\right),\\
		s^{\dagger}_{k,r} &= -(f^{(r)}+2d_0^{(r)})s^{\dagger}_{k-2,r}  +\frac{1}{2}u_{0}^{(r)}\left(\sum_{i+j= k-2}s^{\dagger}_{i,r}s^{\dagger}_{j,r}\right)-\frac{1}{2}u_{0}^{(r)}v_{0}^{(r)}\left(\sum_{i+j= k-3}s^{\dagger}_{i,r}s^{\dagger}_{j,r}\right),\quad k\ge 3.
	\end{aligned}
	\label{rec-dagger}
\end{equation}
Finally, the solution to the problem \eqref{initial5} on the bilateral Somos-5 can be summarized in the following theorem.
\begin{theorem}
The solution $\{\tau_n\}_{n \in \mathbb{Z}}$ to the problem \eqref{initial5} on the  bilateral  Somos-5 recurrence can be given by the Hankel determinants with $\{s_{n,r},s_{n,r}^{\dagger}\}_{n \in \mathbb{N},r = 1,2}$ in \eqref{rec-nodagger}-\eqref{rec-dagger} as elements, that is,
\begin{equation} \label{post_somos5}
	\tau_{2n}=\Delta^{(1)}_{n} =
	\left\{
	\begin{aligned}
		&\det(s_{i+j-2,1})_{1 \le i,j\le  n}, \quad && n \ge 1, \\
		&1, \quad  &&n= 0,\\ 
		&\det(s^{\dagger}_{i+j-2,1})_{1 \le i,j\le  -n}, \quad && n \le -1, \\
	\end{aligned}     
	\right.
\end{equation}
\begin{equation}
	\tau_{2n+1}=\Delta^{(2)}_{n} =
	\left\{
	\begin{aligned}
		&\det(s_{i+j-2,2})_{1 \le i,j\le  n}, \quad && n \ge 1, \\
		&1, \quad  &&n= 0,\\ 
		&\det(s^{\dagger}_{i+j-2,2})_{1 \le i,j\le  -n}, \quad && n \le -1. \\
	\end{aligned}   
	\right.
\end{equation}
\end{theorem}

\begin{remark}
	It seems that the Laurent property for the bilateral Somos-5 recurrence cannot be seen directly from the expressions given above. However, based on the Laurent property for the Somos-4 recurrence, one can at least conclude that 
	\begin{equation*}
		\tau_n \in \mathbb{Z}[\alpha,\beta,\tau_{-2}^{\pm1},\tau_{-1}^{\pm1},\tau_{0}^{\pm1},\tau_{1}^{\pm1},\tau_{2}^{\pm1},\tau_{3}^{\pm1},\tau_{4}^{\pm1},\tau_{5}^{\pm1}],
	\end{equation*}
    and simultaneously
    \begin{equation*}
    	\tau_n \in \mathbb{Z}[\alpha,\beta,\tau_{-5}^{\pm1},\tau_{-4}^{\pm1},\tau_{-3}^{\pm1},\tau_{-2}^{\pm1},\tau_{-1}^{\pm1},\tau_{0}^{\pm1},\tau_{1}^{\pm1},\tau_{2}^{\pm1}].
    \end{equation*}
	Observe that, for the bilateral Somos-5 recurrence, $\{\tau_3,\tau_4,\tau_5\}$ and $\{\tau_{-3},\tau_{-4},\tau_{-5}\}$  can be expressed as Laurent polynomials in $\{\tau_{-2},\tau_{-1},$ $\tau_{0},\tau_{1},\tau_{2}\}$ and the numerators in the rational expressions   are pairwise coprime. This implies
	\begin{equation*}
		\begin{aligned}
			\tau_n \in &\mathbb{Z}[\alpha,\beta,\tau_{-5}^{\pm1},\tau_{-4}^{\pm1},\tau_{-3}^{\pm1},\tau_{-2}^{\pm1},\tau_{-1}^{\pm1},\tau_{0}^{\pm1},\tau_{1}^{\pm1},\tau_{2}^{\pm1}] \\
			&\quad \cap \mathbb{Z}[\alpha,\beta,\tau_{-2}^{\pm1},\tau_{-1}^{\pm1},\tau_{0}^{\pm1},\tau_{1}^{\pm1},\tau_{2}^{\pm1},\tau_{3}^{\pm1},\tau_{4}^{\pm1},\tau_{5}^{\pm1}] \\
			=&\mathbb{Z}[\alpha,\beta,\tau_{-2}^{\pm1},\tau_{-1}^{\pm1},\tau_{0}^{\pm1},\tau_{1}^{\pm1},\tau_{2}^{\pm1}],
		\end{aligned}
	\end{equation*}
	which  reveals the Laurent property for the bilateral Somos-5 recurrence.
\end{remark}

\begin{remark}  It is noted that what we are considering is the Hankel determinant solution to the bilateral Somos-5 with $\tau_{0}=1, \tau_{1}=1, \tau_{2}=x,\tau_{3}=y,\tau_{4}=z$.  Although Hankel determinant formulae to half of the Somos-5  with $\tau_{-2}=\tau_{-1}=1,\tau_{0}=x,\tau_1=y,\tau_2=z$ have also been studied in \cite{chang2015hankel}, the formula there are different from that in \eqref{post_somos5}.  Besides, we also point out that, very recently in \cite{hone2023family}, Hone, Roberts and Vanhaecke  obtained simpler Hankel determinant formulae for the bilateral Somos-5 sequence with initials and coefficients 1s in a nice way.
\end{remark}

\subsection{A symmetric part for the Somos' original conjecture for Somos-4} \label{dualsomos}
In the year 2000, Michael Somos examined the following curve
\begin{equation}
	\mathcal{C}_0 :y-y^2 = z- z^3.
	\label{somoscurve}
\end{equation}
He conjectured that if one expands $\frac{y}{z}-1$ into a power series of $z$ when $z \to 0$, that is,
\begin{equation*}
	\begin{aligned}
		\frac{y}{z}-1 &= \frac{\frac{1}{2}-\sqrt{z^3-z+\frac{1}{4}}}{z}-1 \\
		&= 1z^1+1z^2+3z^3+8z^4+23z^5+68z^6+207z^7+644z^8+2040z^9+\cdots\\
				&= s_0z^1+s_1z^2+s_2z^3 +\cdots,
	\end{aligned}
\end{equation*}
where 
        $s_k = 2s_{k-1}+\sum_{i+j=k-2}s_is_j,k\ge 2$ with         $s_0 = 
        s_1 = 1$,
then the corresponding sequence of Hankel determinants $\{\Delta_{n}\}_{n \in \mathbb{N}}$ defined by 
\begin{equation*}
	\Delta_{n} := \det(s_{i+j-2})_{1 \le i.j \le n}
\end{equation*}
satisfies half of the Somos-4 recurrence
\begin{equation}
	\Delta_{n+4}\Delta_{n}=\Delta_{n+3}\Delta_{n+1}+(\Delta_{n+2})^2,\quad \forall\, n \in \mathbb{N}.
	\label{somoe}
\end{equation}
This conjecture can be proved using various techniques in \cite{xin2009proof} and \cite{chang2012conjecture}. 
\par

{It is noted that our scheme works for this particular curve. Indeed, if we take the birational transformation 
\begin{equation*}
    \begin{aligned}
        y &= \frac{1+(X+1)^{-2}Y}{2},\qquad z = (X+1)^{-1},
    \end{aligned}
\end{equation*}
the curve \eqref{somoscurve} becomes
\begin{equation*}
    Y^{2} =  (X^2 -3)^2 - 4(X+2),
\end{equation*}
which is exactly in the form of the spectral curve \eqref{curve} when the parameters are taken to be $f= -3,u=-1,v= 2$. And the generating function $\frac{y}{z} - 1= \frac{1}{2(X+1)} Y + \frac{X^2-1}{2(X+1)}$ 
corresponds to the choice of $Y_0$ when the parameters take the values $u_0 = 2, v_0 = -1, d_0 = 1$. Also see \cite[Theorem 5.3, Remark 5.4]{hone2021}.
The action of the Galois group $Gal(\mathcal{F}/\mathbb{C}(X))$ is independent of coordinates, therefore we are free to apply our results to the original coordinates $(y,z)$.}
\par
By using our methods based on quadratic orthogonal pairs, we can find another curve whose Taylor series expansion will give the negative part of the same Somos-4 recurrence.
{Indeed, let $\frac{\Tilde{y}}{z} - 1 $ be the generating function for the corresponding negative sequence, by relation \eqref{symmetric}, we have
\begin{equation*}
    (1-\frac{\Tilde{y}}{z})(\frac{y}{z}- 1)^* = -1.
\end{equation*}
Substituting $\frac{y}{z}- 1= \frac{\frac{1}{2}-\sqrt{z^3-z+\frac{1}{4}}}{z}-1$ into it and using the definition for the symbol $*$, we immediately find the algebraic relation for $\Tilde{y}$ as:}
\begin{equation*}
	\Tilde{\mathcal{C}}_0 :(1-z)\Tilde{y}^2+(2z^2-1)\Tilde{y} = z(z^2-1),
\end{equation*}
{which is the exact equation for the dual curve.} If one expands $\frac{\Tilde{y}}{z}-1$ into a power series of $z$ when $z \to 0$, one has
\begin{equation*}
	\begin{aligned}
		\frac{\Tilde{y}}{z}-1 &= \frac{\sqrt{z^3-z+\frac{1}{4}}+z-\frac{1}{2}}{z(z-1)}\\
		&= 1z^1+2z^2+5z^3+13z^4+36z^5+104z^6+311z^7+955z^8+2995z^9+\cdots\\
				&= \Tilde{s}_0z^1+\Tilde{s}_1z^2+\Tilde{s}_2z^3 +\cdots,
	\end{aligned}
\end{equation*}
where $
        \Tilde{s}_{k} =3\Tilde{s}_{k-1}-2\Tilde{s}_{k-3} + \sum_{i+j = k-3}(\Tilde{s}_{i}-\Tilde{s}_{i+1}) (\Tilde{s}_{j}-\Tilde{s}_{j+1}),k\ge 3$ 
   with $\Tilde{s}_0= 1, \Tilde{s}_1 = \Tilde{s}_1 = 2 $.
     Define
\begin{equation*}
	\tau_{n} =
	\left\{
	\begin{aligned}
		&\det(s_{i+j-2})_{1 \le i,j\le  n}, \quad && n \ge 1, \\
		&1, \quad  &&n= 0,\\ 
		&\det(\Tilde{s}_{i+j-2})_{1 \le i,j\le  -n}, \quad && n \le -1, \\
	\end{aligned}       
	\right.
\end{equation*}
then $\{\tau_n\}_{n \in \mathbb{Z}}$ is a bilateral Somos-4 sequence satisfying 
\begin{equation*}
	\tau_{n+4}\tau_{n}=\tau_{n+3}\tau_{n+1}+(\tau_{n+2})^2,\quad \forall\, n \in \mathbb{Z},
\end{equation*}
with $\tau_{-2}=\tau_{-1}=\tau_0=\tau_1=1$.

\section{Acknowledgement}
We thank Professor A.N.W. Hone for his valuable communications.
X.K. Chang's research was supported by the National Natural Science Foundation of China (Grant Nos. 12288201, 12222119) and the Youth Innovation Promotion Association CAS.

\appendix

\section{Quadratic orthogonal pairs for general hyperelliptic curves} \label{generalizepairs}
In this appendix let's discuss quadratic orthogonal pairs for general hyperelliptic curves. Consider a hyperelliptic curve of genus $g$ defined by 
\begin{equation}
	\mathcal{C}:Y^2 = A(X)^2 + 4R(X),
\end{equation}
where $g$ is an arbitrary positive integer and $A(X)$,$R(X)$ are polynomials defined by
\begin{equation*}
	\begin{aligned}
		A(X) &= X^{g+1} +\sum_{j=0}^{g-1}k^{(j)}X^j,\qquad R(X) &= u\left(X^{g}+ \sum_{j=0}^{g-1}l^{(j)}X^j\right).
	\end{aligned}
\end{equation*}
\par
In \cite{hone2021}, Hone constructed two sequences of polynomials $\{P_n\}_{n \in \mathbb{Z}}$ and $\{Q_n\}_{n \in \mathbb{Z}}$ satisfying the following iteration
\begin{equation}
	\begin{aligned}
		P_{n+1} = a_{n}Q_{n} - P_{n},\qquad Q_{n+1} = Q_{n-1} +2a_{n}P_{n} - a_{n}^2Q_{n}, \qquad \forall\, n \in \mathbb{Z},
	\end{aligned} \label{iteration}
\end{equation}
where $\{P_n\}_{n \in \mathbb{Z}}$ and $\{Q_n\}_{n \in \mathbb{Z}}$ are expressed by
\begin{equation*}
	\begin{aligned}
		P_{n}(X) = A(X)+ \sum_{j=0}^{g-1}\pi_{n}^{(j)}X^j,\qquad Q_{n}(X) = u_{n}\left(X^g+\sum_{j=0}^{g-1}\rho_{n}^{(j)}X^j\right),
	\end{aligned}
\end{equation*}
and $\{a_n\}_{n \in \mathbb{Z}}$ is a sequence of linear functions uniquely determined by the coefficients of $\{Q_{n}\}_{n \in \mathbb{Z}}$ via
\begin{equation*}
	a_{n}(X) = \frac{2}{u_n}(X-\rho_{n}^{(g-1)}).
\end{equation*}
By use of the iteration \eqref{iteration}, it is straightforward to compute
\begin{equation*}
	\begin{aligned}
		P_{n+1}^2 + Q_{n+1}Q_{n}  = (a_{n}Q_{n} - P_{n})^2+(Q_{n-1} +2a_{n}P_{n} - a_{n}^2Q_{n})Q_{n}  = P_{n}^2+Q_{n}Q_{n-1},
	\end{aligned}
\end{equation*}
which shows that $P_{n}^2+Q_{n}Q_{n-1}$ is a conserved quantity for the iteration \eqref{iteration}. Therefore if we choose initial values $\{P_0,Q_0,Q_{-1}\}$ such that 
\begin{equation*}
	P_{0}^2+Q_{0}Q_{-1} = Y^2,
\end{equation*}
then we will have
\begin{equation}
	P_{n}^2+Q_{n}Q_{n-1} = Y^2 \label{conserved}
\end{equation}
for any integer $n$.
\par
The above two sequences of polynomials $\{P_n\}_{n \in \mathbb{Z}}$ and $\{Q_n\}_{n \in \mathbb{Z}}$ give rise to a choice of sequences of hyperelliptic functions. To be precise, let us choose one sequence $\{Y_n=r_n Y+s_n|r_n,s_n\in \mathbb{C}(X),r_n \ne 0 \}_{n\in \mathbb{Z}}$ and $\{a_n\}_{n \in \mathbb{Z}}$  by setting
\begin{equation}
	\begin{aligned}
		&r_n := \frac{1}{Q_{n}},\qquad s_n :=\frac{P_n}{Q_{n}},\qquad a_n := a_n.
	\end{aligned}
	\label{choiceg1}
\end{equation}
and the other sequence $\{\Tilde{Y}_n=\Tilde{r}_nY+\Tilde{s}_n|\Tilde{r}_n,\Tilde{s}_n\in \mathbb{C}(X),\Tilde{r}_n \ne 0 \}_{n\in \mathbb{N}}$ and $\{\Tilde{a}_n\}_{n \in \mathbb{N}}$ by setting
\begin{equation}
	\begin{aligned}
		&\Tilde{r}_{n} := \frac{1}{Q_{n-1}},\qquad \Tilde{s}_{n} :=\frac{P_n}{Q_{n-1}},\qquad \Tilde{a}_{n} :=a_{n-1}.
	\end{aligned}
	\label{choiceg2}
\end{equation}
As a result, we can obtain generalizations of  Lemma \ref{Ylemma0}, Lemma \ref{Ylemma1} and Theorem \ref{pair}.
\begin{lemma}	\label{Ylemma0G}
	Given $\{P_0,Q_0,Q_{-1}\}$ satisfying $P_{0}^2+Q_{0}Q_{-1} = Y^2$, the sequences $\{Y_n=r_nY+s_n|r_n,s_n\in \mathbb{C}(X),r_n \ne 0 \}_{n\in \mathbb{Z}}$ and $\{a_n\}_{n \in \mathbb{Z}}$ defined according to \eqref{choiceg1} admit the following properties:
		\begin{enumerate}[(i)]
		\item $Y_n= a_n+\frac{1}{Y_{n+1}},\quad \forall\, n \in \mathbb{Z}$;
		\item $a_n = \left \lfloor Y_n \right \rfloor,\quad \forall\, n \in \mathbb{Z}$;
		\item $Y_{n}=\langle a_n,a_{n+1},\ldots \rangle, \quad \forall\, n \in \mathbb{Z} $;
		\item $Y_n$ is J-expressible and Lax representable for any integer $n$.
	\end{enumerate}
\end{lemma}
\begin{proof}
	The proof for (i) follows immediately from \eqref{iteration} and \eqref{conserved}. To prove (ii), it suffices to show
	\begin{equation*}
		Y_{n}- a_{n} \in \mathcal{O}(X^{-1}),
	\end{equation*}
which follows from
	\begin{equation*}
		\begin{aligned}
			Y_{n}- a_{n}  &= \frac{Y+P_{n}-a_{n}Q_{n}}{Q_{n}} \\
			& = \frac{Y-P_{n+1}}{Q_{n}} \\
			& = \frac{\sqrt{A^2+4R}-A}{u_{n}\left(X^g+\sum_{j=0}^{g-1}\rho_{n}^{(j)}X^j\right)}-\frac{\sum_{j=0}^{g-1}\pi_{n+1}^{(j)}X^j}
			{u_{n}\left(X^g+\sum_{j=0}^{g-1}\rho_{n}^{(j)}X^j\right)} \\
			& = \frac{4R}{\left(\sqrt{A^2+4R}+A\right)\left(u_{n}\left(X^g+\sum_{j=0}^{g-1}\rho_{n}^{(j)}X^j\right)\right)}-\frac{\sum_{j=0}^{g-1}\pi_{n+1}^{(j)}X^j}
			{u_{n}\left(X^g+\sum_{j=0}^{g-1}\rho_{n}^{(j)}X^j\right)} \\
			&\in \mathcal{O}(X^{-1}).
		\end{aligned}
	\end{equation*}
	Combining (i) and (ii), we immediately get (iii). Furthermore, it is not hard to arrive at (iv) by use of (iii) and \eqref{conserved}.
\end{proof}

\begin{lemma}	\label{Ylemma1G}
	Given $\{P_0,Q_0,Q_{-1}\}$ satisfying $P_{0}^2+Q_{0}Q_{-1} = Y^2$, the sequences $\{\Tilde{Y}_n=\Tilde{r}_nY+\Tilde{s}_n|\Tilde{r}_n,\Tilde{s}_n\in \mathbb{C}(X),\Tilde{r}_n \ne 0 \}_{n\in \mathbb{N}}$ and $\{\Tilde{a}_n\}_{n \in \mathbb{N}}$ defined according to \eqref{choiceg2} admit the following properties:
		\begin{enumerate}[(i)]
		\item $\Tilde{Y}_{n}= \Tilde{a}_{n}+\frac{1}{\Tilde{Y}_{n-1}},\quad \forall\, n \in \mathbb{Z}$;
		\item $\Tilde{a}_{n} = \left \lfloor \Tilde{Y}_{n} \right \rfloor,\quad \forall\, n \in \mathbb{Z}$;
		\item $\Tilde{Y}_{n}= \langle \Tilde{a}_{n},\Tilde{a}_{n-1},\ldots \rangle, \quad \forall\, n \in \mathbb{Z}  $;
		\item $\Tilde{Y}_{n}$ is J-expressible and Lax representable for any integer $n$.
	\end{enumerate}
\end{lemma}
\begin{proof}
	The proof of (i), (iii) and (iv) can be achieved in the same way as those in the previous lemma. To prove (ii), we observe that
\begin{equation*}
		\begin{aligned}
			\left \lfloor \Tilde{Y}_{n} \right \rfloor & = \left \lfloor \frac{Y+P_{n}}{Q_{n-1}} \right \rfloor \\
			& = \left \lfloor \frac{Y+P_{n-1}}{Q_{n-1}} \right \rfloor + \left \lfloor \frac{P_{n}-P_{n-1}}{Q_{n-1}} \right \rfloor \\
			& = \left \lfloor Y_{n-1} \right \rfloor + \left \lfloor \frac{\sum_{j=0}^{g-1}\pi_{n}^{(j)}X^j-\sum_{j=0}^{g-1}\pi_{n-1}^{(j)}X^j}{Q_{n-1}} \right \rfloor \\
			& = a_{n-1}.
		\end{aligned}
	\end{equation*}
\end{proof}

\begin{theorem}	\label{pairG}
	Given $\{P_0,Q_0,Q_{-1}\}$ satisfying $P_{0}^2+Q_{0}Q_{-1} = Y^2$, the two pairs of sequences $\{Y_n=r_nY+s_n|r_n,s_n\in \mathbb{C}(X),r_n \ne 0 \}_{n\in \mathbb{Z}}$, $\{a_n\}_{n \in \mathbb{Z}}$ and $\{\Tilde{Y}_n=\Tilde{r}_nY+\Tilde{s}_n|\Tilde{r}_n,\Tilde{s}_n\in \mathbb{C}(X),\Tilde{r}_n \ne 0 \}_{n\in \mathbb{N}}$ and $\{\Tilde{a}_n\}_{n \in \mathbb{N}}$ defined by \eqref{choiceg1} and \eqref{choiceg2} satisfy
	\begin{enumerate}[(i)]
		\item $\Tilde{a}_{n}=a_{n-1},\quad \forall\, n \in \mathbb{Z}$;
		\item  for any integer $n$, $(Y_n,\Tilde{Y}_n)$ is a quadratic orthogonal pair, that is, $\tilde Y_nY_n^*=-1$.
	\end{enumerate}
\end{theorem}
\begin{proof}
	The conclusion immediately follows  from the explicit construction \eqref{choiceg1} and \eqref{choiceg2} and the conserved quantity \eqref{conserved}.
\end{proof}

{With Lemma \ref{Ylemma0G}, Lemma \ref{Ylemma1G} and Theorem \ref{pairG}, one can go through Section \ref{sec:Hankelexpressions} and Section \ref{sec:connect} in the case of any genus $g$ and obtain the generalized version of all results of Section \ref{sec:Hankelexpressions} and Theorem \ref{rela_n-n-1} and Corollary \ref{coro_tau} without much effort. However, Lemma \ref{lemmad} and Corollary \ref{corosomos4} will be replaced by a much more complicated version that is still unknown in general. Even so, the notion of quadratic orthogonal pairs is still valid in resolving the mismatch problem. As an example to show how this works in higher genus, we reconsider the curve of genus 2 examined in \cite[Example 4.3 and 4.5]{hone2021} by Hone:
\begin{equation*}
        Y^2 = (X^3-5X-1)^2-4(X^2+2X+3).
\end{equation*}
 The generating function for the positive half sequence he picked on this curve is denoted by $G$ whose explicit formula and series expansion at $X \to \infty$ and $Y \sim  X^3$ are given by
\begin{equation*}
    \begin{aligned}
        G&=\frac{X+\frac{1}{2}+\frac{1}{4}(X^3-5X-1-Y)}{X^2+\frac{1}{2}X-\frac{3}{2}} \\
        &= X^{-1} + 0X^{-2}+ 2X^{-3} + 0X^{-4}+7X^{-5} + 2X^{-6} + \cdots, \\
        &= s_0 X^{-1} + s_1 X^{-2} + s_2 X^{-3} + \cdots,
    \end{aligned}
\end{equation*}
and the moments $\{s_{j}\}_{j \in \mathbb{N}}$ are recursively produced by
\begin{equation*}
    s_j = s_{j-2} - s_{j-3} + 2 \sum_{i=0}^{j-2} s_i s_{j-2-i} + \sum_{i=0}^{j-3} s_i s_{j-3-i} - 3 \sum_{i=0}^{j-4} s_i s_{j-4-i}, \quad j \geq 3,
\end{equation*}
with initial values \(s_0 = 1\), \(s_1 = 0\), \(s_2 = 2\). The corresponding sequence of Hankel determinants determined by $G$ now reads:
\begin{equation}  \label{somos8positive}
    1,1,2,6,31,319,5810,147719,8526736,\ldots.
\end{equation}
It was shown in \cite[Theorem 5.5, Example 5.6]{hone2021} with the use of computer algebra that this sequence is the positive half of a particular Somos-8 sequence:
\begin{equation*}
    7\tau_{n+8}\tau_n + 137\tau_{n+7}\tau_{n+1} + 2504\tau_{n+6}\tau_{n+2} - 43424\tau_{n+5}\tau_{n+3} - 26959\tau_{n+4}^2 = 0.
\end{equation*}
To obtain the negative half sequence, he picked another function denoted by $G^{\dagger}$ whose explicit formula and series expansion at $X \to \infty$ and $Y \sim  X^3$ is given by
\begin{equation*}
    \begin{aligned}
        G^{\dagger} &= -\frac{1}{4}\left( \frac{-Y+X^3-\frac{5}{2}X+\frac{1}{2}}{X^2+\frac{1}{2}X -\frac{3}{2}} \right) \\
        &= -\frac{5}{8}X^{-1} - \frac{1}{16}X^{-2} - \frac{45}{32}X^{-3} - \frac{25}{64}X^{-4} - \frac{757}{128}X^{-5} - \frac{801}{256}X^{-6} - \cdots, \\
        &= s_0^{\dagger} X^{-1} + s_1^{\dagger} X^{-2} + s_2^{\dagger} X^{-3} + \cdots,
    \end{aligned}
\end{equation*}
as in \cite[Example 4.5]{hone2021}. Here we use $-Y$ to replace the $Y$ there for consistency. The recursion for the moments $\{s^{\dagger}_{j}\}_{j \in \mathbb{N}}$ is:
\begin{equation*}
    s_j^\dagger = \frac{5}{2} s_{j-2}^\dagger - \frac{1}{2} s_{j-3}^\dagger - 2 \sum_{i=0}^{j-2} s_i^\dagger s_{j-2-i}^\dagger - \sum_{i=0}^{j-3} s_i^\dagger s_{j-3-i}^\dagger + 3 \sum_{i=0}^{j-4} s_i^\dagger s_{j-4-i}^\dagger, \quad j \geq 3,
\end{equation*}
with initial values \( s_0^\dagger = -5/8, s_1^\dagger = -1/16, s_2^\dagger = 45/32 \). The corresponding sequence of Hankel determinants determined by $G^{\dagger}$ reads:
\begin{equation}  \label{somos8negative}
    1,-\frac{5}{2^3}, \frac{7}{2^3}, -\frac{303}{2^7}, \frac{4091}{2^9}, -\frac{63805}{2^{10}}, \frac{3496637}{2^{12}}, \ldots.
\end{equation}
Clearly, one cannot simply put the sequences \eqref{somos8positive} and \eqref{somos8negative} together to derive a bilateral Somos-8 sequence, therefore a gauge transformation of the form \cite[ (4.31)]{hone2021} is needed. However, based on the notion of quadratic orthogonal pairs that $(-\Tilde{G})^{-1} (G^{-1})^*=-1$, it is rather easy to find the precise generating function $\Tilde{G}$ whose sequence of Hankel determinants gives exactly the corresponding negative sequence. More precisely, we have
\begin{equation*}
    \begin{aligned}
        \Tilde{G} = \frac{1}{G^*} &= \frac{X^2+\frac{1}{2}X-\frac{3}{2}}{X+\frac{1}{2}+\frac{1}{4}(X^3-5X-1+Y)} \\
        &= 2 X^{-1}+ 1 X^{-2} + 3X^{-3} + 3X^{-4} + 11 X^{-5}+ 14X^{-6} + 54X^{-7} + 81X^{-8} + 300 X^{-9}+ \cdots,\\
        &= \Tilde{s}_{0} X^{-1} + \Tilde{s}_{1} X^{-2} + \Tilde{s}_{2} X^{-3} +\cdots,
    \end{aligned}
\end{equation*}
where the moments $\{\Tilde{s}_{j}\}_{j \in \mathbb{N}}$ are recursively produced by
\begin{equation*}
    \Tilde{s}_j = \Tilde{s}_{j-2} - \Tilde{s}_{j-3} + \sum_{i=0}^{j-2} \Tilde{s}_i \Tilde{s}_{j-2-i}  -\sum_{i=0}^{j-4} \Tilde{s}_i \Tilde{s}_{j-4-i}, \quad j \geq 3,
\end{equation*}
with initial values $ \Tilde{s}_0 = 2, \Tilde{s}_1 = 1, \Tilde{s}_2 = 3 $, and the corresponding sequence of Hankel determinants determined by $\Tilde{G}$ reads:
\begin{equation}  \label{somos8negativeour}
1,2,5,28,303,4091,127610,6993274,562196701,89628659609,26028173673848, \ldots.
\end{equation}
Combining the sequences \eqref{somos8positive} and \eqref{somos8negativeour} will exactly yield the doubly infinite sequence presented in \cite[(4.32)]{hone2021}, that is,
\begin{equation} \label{somos8bilateral}
    \begin{aligned}     &\ldots,26028173673848,89628659609,562196701,6993274,127610,4091,\\&\qquad\qquad\qquad\qquad\qquad\qquad\qquad\,303,28,5,2,1, 1,2,6,31,319,5810,147719,8526736,\ldots.
    \end{aligned}
\end{equation}
It is also noted that, using our recursion for the moments $\{s_{j}\}_{j \in \mathbb{N}}$ and $\{\Tilde{s}_{j}\}_{j \in \mathbb{N}}$, one can immediately see that the doubly infinite sequence \eqref{somos8bilateral} consists entirely of integers. This is not obvious when using the recursion for $\{s^{\dagger}_{j}\}_{j \in \mathbb{N}}$ and applying the gauge transformation.
}
\small
\bibliographystyle{abbrv}
\bibliographystyle{plain}

\begin{thebibliography}{10}

\bibitem{adams1980}
{
W. Adams and M. Razar. Multiples of points on elliptic curves and continued fractions. \newblock {\em Proc. London
Math. Soc.}, 41:481--498, 1980.
}


\bibitem{barry2010generalized}
P.~Barry.
\newblock Generalized {C}atalan numbers, {H}ankel transforms and {S}omos-4
  sequences.
\newblock {\em J. Integer Seq.}, 13(7), 2010.

\bibitem{carroll2004cube}
G.~Carroll and D.~Speyer.
\newblock The cube recurrence.
\newblock {\em Electron. J. Combin.}, 11:R73, 2004.

\bibitem{chang2012conjecture}
X.~Chang and X.~Hu.
\newblock A conjecture based on {S}omos-4 sequence and its extension.
\newblock {\em Linear Algebra Appl.}, 436(11):4285--4295, 2012.

\bibitem{chang2015hankel}
X.~Chang, X.~Hu, and G.~Xin.
\newblock Hankel determinant solutions to several discrete integrable systems
  and the {L}aurent property.
\newblock {\em SIAM J. Discrete Math.}, 29(1):667--682, 2015.

\bibitem{eager2012colored}
R.~Eager and S.~Franco.
\newblock {Colored BPS pyramid partition functions, quivers and cluster
  transformations}.
\newblock {\em J. High Energy Phys.}, 2012(9):1--44, 2012.

\bibitem{everest2003recurrence}
G.~Everest, A.~van der~Poorten, I.~Shparlinski, and T.~Ward.
\newblock {\em Recurrence sequences}, volume 104.
\newblock Mathematical Surveys and Monographs, 104. American Mathematical
  Society, Providence, R.I., 2003.

\bibitem{everest2005introduction}
G.~Everest and T.~Ward.
\newblock {\em An introduction to number theory}.
\newblock Springer, Berlin, 2005.

\bibitem{fedorov2016sigma}
Y.~N. Fedorov and A.~Hone.
\newblock {Sigma-function solution to the general Somos-6 recurrence via
  hyperelliptic Prym varieties}.
\newblock {\em J. Integrable Systems}, 1(1):xyw012, 2016.

\bibitem{fomin2002cluster}
S.~Fomin and A.~Zelevinsky.
\newblock Cluster algebras {I}: Foundations.
\newblock {\em J. Amer. Math. Soc.}, 15(2):497--529, 2002.

\bibitem{fomin2002laurent}
S.~Fomin and A.~Zelevinsky.
\newblock The {L}aurent {P}henomenon.
\newblock {\em Adv. Appl. Math.}, 28(2):119--144, 2002.

\bibitem{Gale1991strange}
D.~Gale.
\newblock The strange and surprising saga of the {S}omos sequences.
\newblock {\em Math. Intelligencer}, 13(1):40--42, 1991.

\bibitem{han2025}
G. Han and E. Pedon. Hankel continued fractions and Hankel determinants for $ q $-deformed metallic numbers. arXiv:2502.05993, 2025.

\bibitem{hone2005elliptic}
A.~Hone.
\newblock Elliptic curves and quadratic recurrence sequences.
\newblock {\em Bull. London Math. Soc.}, 37(02):161--171, 2005.

\bibitem{hone2007sigma}
A.~Hone.
\newblock Sigma function solution of the initial value problem for {S}omos 5
  sequences.
\newblock {\em Trans. Amer. Math. Soc.}, 359(10):5019--5035, 2007.

\bibitem{hone2010analytic}
A.~Hone.
\newblock Analytic solutions and integrability for bilinear recurrences of
  order six.
\newblock {\em Appl. Anal.}, 89(4):473--492, 2010.

\bibitem{hone2021}
A.~Hone.
\newblock {Continued fractions and Hankel determinants from hyperelliptic
  curves}.
\newblock {\em Commun. Pure Appl. Math.}, 74(11):2310--2347, 2021.

\bibitem{hone2022heron}
A.~Hone.
\newblock {Heron triangles with two rational medians and Somos-5 sequences}.
\newblock {\em European J. Math.}, 8(4):1424--1486, 2022.

\bibitem{hone2023casting}
A.~Hone.
\newblock Casting light on shadow \uppercase{s}omos sequences.
\newblock {\em Glasgow Math. J.}, 65(S1):S87--S101, 2023.

\bibitem{hone2018some}
A.~Hone, T.~Kouloukas, and G.~Quispel.
\newblock {Some integrable maps and their Hirota bilinear forms}.
\newblock {\em J. Phys. A: Math. Theor.}, 51(4):044004, 2018.

\bibitem{hone2017reductions}
A.~Hone, T.~Kouloukas, and C.~Ward.
\newblock {On reductions of the Hirota-Miwa equation}.
\newblock {\em SIGMA. Symmetry, Integrability and Geom. Methods Appl.}, 13:057,
  2017.

\bibitem{hone2023family}
A.~Hone, J.~Roberts, and P.~Vanhaecke.
\newblock {A family of integrable maps associated with the Volterra lattice}.
\newblock {\em Nonlinearity}, 37(9):095028, 2024.

\bibitem{krattenthaler2005advanced}
C.~Krattenthaler.
\newblock Advanced determinant calculus: a complement.
\newblock {\em Linear Algebra Appl.}, 411:68--166, 2005.

\bibitem{krichever1998elliptic}
I. Krichever, P. Wiegmann and A. Zabrodin.
 Elliptic solutions to difference non-linear equations and related many-body problems.
{\em Comm. Math. Phys.}, 193:373--396, 1998.

\bibitem{malouf1992integer}
J.~Malouf.
\newblock An integer sequence from a rational recursion.
\newblock {\em Discrete Math.}, 110(1-3):257--261, 1992.

\bibitem{ovsienko2025continued}
V. Ovsienko and E. Pedon. Continued fractions for $q$-deformed real numbers, $\{-1, 0, 1\}$-Hankel determinants, and Somos--Gale--Robinson sequences. Adv.  Appl. Math., 162:102788, 2025.

\bibitem{quispel1988integrable}
G.~Quispel, J.~Roberts, and C.~Thompson.
\newblock Integrable mappings and soliton equations.
\newblock {\em Phys. Lett. A}, 126(7):419--421, 1988.

\bibitem{robinson1992periodicity}
R.~Robinson.
\newblock Periodicity of {S}omos sequences.
\newblock {\em Proc. Amer. Math. Soc}, 116(3):613--619, 1992.

\bibitem{somos1989problem}
M.~Somos.
\newblock Problem 1470.
\newblock {\em Crux Mathematicorum}, 15:208, 1989.

\bibitem{speyer2007perfect}
D.~Speyer.
\newblock Perfect matchings and the octahedron recurrence.
\newblock {\em J. Algebraic Combin.}, 25(3):309--348, 2007.

\bibitem{swart2003elliptic}
C.~Swart.
\newblock {\em Elliptic curves and related sequences}.
\newblock PhD thesis, University of London, 2003.

\bibitem{van3}
A.J. van der~Poorten.
\newblock Hyperelliptic curves, continued fractions, and {S}omos sequences.
\newblock In {\em Dynamics $\&$ stochastics}, pages 212--224. IMS Lecture Notes
  Monogr. Ser., 48. Inst. Math. Statist., Beachwood, OH, 2006.

\bibitem{van2}
A.J. van der~Poorten.
\newblock Elliptic curves and continued fractions.
\newblock {\em J. Integer Seq.}, 8(2):Article 05.2.5, 19 pp., 2005.

\bibitem{van1}
A.J. van der~Poorten.
\newblock Determined sequences, continued fractions, and hyperelliptic curves.
\newblock In {\em International Algorithmic Number Theory Symposium}, pages
  393--405. Springer, 2006.
  
  \bibitem{van2006recurrence}
   A.J. van der Poorten and C.S. Swart. Recurrence relations for elliptic sequences: every Somos 4 is a Somos-$k$, Bull. Lond. Math. Soc. 38:546--554, 2006.
  
  \bibitem{wang2024sufficient}
  Y. Wang and Z. Zhang. Sufficient condition for $(\alpha, \beta)$ Somos 4 Hankel determinants. \newblock {\em Discrete Math.}, 347: 113937, 2024.

\bibitem{wall1}
H.~Wall.
\newblock Note on the expansion of a power series into a continued fraction.
\newblock {\em Bull. Amer. Math. Soc.}, 51(11):97--105, 1945.

\bibitem{wall2}
H.~Wall.
\newblock {\em Analytic theory of continued fractions}.
\newblock Van Nostrand, New York, 1948.

\bibitem{ward1948memoir}
M.~Ward.
\newblock Memoir on elliptic divisibility sequences.
\newblock {\em Amer. J. Math.}, 70(1):31--74, 1948.

\bibitem{xin2009proof}
G.~Xin.
\newblock Proof of the {S}omos-$4$ {H}ankel determinants conjecture.
\newblock {\em Adv. Appl. Math.}, 42(2):152--156, 2009.

\bibitem{zannier2019hyper}
{
U. Zannier. Hyperelliptic continued fractions and generalized Jacobians. \newblock {\em Amer. J. Math.}, 141(1):1--40, 2019.
}
\end{thebibliography}

\end{document}